\documentclass[11pt,a4paper]{article}
\usepackage{vmargin}
\setmarginsrb{1.0in}{1.0in}{1.0in}{1.0in}{0mm}{0mm}{5mm}{5mm}
\usepackage[utf8]{inputenc}
\usepackage[T1]{fontenc}
\usepackage{rotating, graphicx}
\usepackage{lscape}

\title{Selfish Creation of Social Networks\\{\small (full version)}}
\author{Davide Bil{\`{o}}\thanks{Department of Humanities and Social Sciences, University of Sassari, Italy} \and Tobias Friedrich\thanks{Hasso Plattner Institute, University of Potsdam, Germany, \texttt{firstname.lastname@hpi.de}} \and Pascal Lenzner\footnotemark[2] \and Stefanie Lowski\thanks{Department of Computer Science, Humboldt-University Berlin, Germany, \texttt{lowski@math.hu-berlin.de}}\and Anna Melnichenko\footnotemark[2] }
\date{~}

\usepackage{helvet} % DO NOT CHANGE THIS
\usepackage{courier}  % DO NOT CHANGE THIS
\usepackage[hyphens]{url}  % DO NOT CHANGE THIS

\usepackage{hyperref}
\usepackage{graphicx}
\usepackage{tabularx} 
\usepackage{xspace}
\usepackage{xcolor}
\usepackage{caption}
\usepackage{multirow}
\usepackage{booktabs} %for tables
\renewcommand{\arraystretch}{1.3} %More space between rows in tables
\usepackage[utf8]{inputenc}
\usepackage{amsfonts, amsthm}
\usepackage{amsmath,amssymb}
\usepackage[switch]{lineno}
\usepackage{footmisc}
\usepackage{pgfplots}
\pgfplotsset{compat=1.10}
\usepgfplotslibrary{statistics}
\usepackage{todonotes}

\newcommand{\PS}{pairwise stable\xspace}
\newcommand{\dist}{{\mathit d}} %distance

\newtheorem{theorem}{Theorem}
\newtheorem{lemma}{Lemma}
\newtheorem{proposition}{Proposition}

\begin{document}
%\linenumbers
\maketitle

\begin{abstract}
\noindent Understanding real-world networks has been a core research endeavor throughout the last two decades. Network Creation Games are a promising approach for this from a game-theoretic perspective. In these games, selfish agents corresponding to nodes in a network strategically decide which links to form to optimize their centrality. Many versions have been introduced and analyzed, but none of them fits to modeling the evolution of social networks. In real-world social networks, connections are often established by recommendations from common acquaintances or by a chain of such recommendations. Thus establishing and maintaining a contact with a friend of a friend is easier than connecting to complete strangers. This explains the high clustering, i.e., the abundance of triangles, in real-world social networks. 

We propose and analyze a network creation model inspired by real-world social networks. Edges are formed in our model via bilateral consent of both endpoints and the cost for establishing and maintaining an edge is proportional to the distance of the endpoints before establishing the connection. We provide results for generic cost functions, which essentially only must be convex functions in the distance of the endpoints without the respective edge. For this broad class of cost functions, we provide many structural properties of equilibrium networks and prove (almost) tight bounds on the diameter, the Price of Anarchy and the Price of Stability. Moreover, as a proof-of-concept we show via experiments that the created equilibrium networks of our model indeed closely mimics real-world social networks. We observe degree distributions that seem to follow a power-law, high clustering, and low diameters. This can be seen as a promising first step towards game-theoretic network creation models that predict networks featuring all core real-world properties.
\end{abstract}

\section{Introduction}
Complex networks from the Internet to various (online) social networks have a huge impact on our lives and it is thus an important research challenge to understand these networks and the forces that shape them. The emergence of the Internet has kindled the interdisciplinary field of Network Science~\cite{Bar16}, which is devoted to analyzing and understanding real-world networks. 

Extensive research, e.g. \cite{RHB99,B99,BKM00,K00,LKF05,NBW11,Bar16}, on real world networks from many different domains like communication networks, social networks, metabolic networks, etc.  
has revealed the astonishing fact that most of these real-world networks share the following basic properties: 
\begin{itemize}
 \item \emph{Small-world property:} The diameter and average distances are at most logarithmic in number of nodes.
 \item \emph{Clustering:} Two nodes with a common neighbor have a high probability of being neighbors, i.e., there is an abundance of triangles and small cliques. 
 \item \emph{Power-law degree distribution:} The probability that a node has degree $k$ is proportional to $k^{-\beta}$, for $2 \leq \beta \leq 3$. That is, the degree distribution follows a power-law. Such networks are also called \emph{scale-free networks}.
\end{itemize}
The phenomenon that real world networks from different domains are very similar begs a scientific explanation, i.e., formal models that generate networks with the above properties from very simple rules.

Many such models have been proposed, most prominently the preferential attachment model~\cite{B99}, Chung-Lu random graphs~\cite{CL02}, hyperbolic random graphs~\cite{Kri10,FK15} and geometric inhomogeneous random graphs~\cite{BKL15}. However, all these models describe a purely random process which eventually outputs a network having realistic properties. In contrast, many real-world networks evolved over time by an interaction of rational agents. For example, in (online) social networks~\cite{Jac10} the selfish agents correspond to people or firms that choose carefully with whom to maintain a connection. Thus, a model with higher explanatory power should consider rational selfish agents which use and modify the network to their advantage~\cite{Pap01}.

In game-theoretic models for network formation, selfish agents correspond to nodes in a network. Each agent strategically selects other agents to form a link. The union of all chosen links then determines the edge-set of the created network~\cite{JW96}. The individual goal of each agent is modeled via a cost function, which typically consists of costs for creating links and of a service cost term, which measures the perceived quality of the created network for the individual agent. For example, the service cost could be proportional to the node's centrality~\cite{Fab03} or simply to the number of reachable nodes~\cite{BG00}. Any network can be considered as an outcome of such a game and among all possible networks the so-called equilibrium networks, where no agent wants to add or remove links, are particularly interesting since analyzing their structure yields insights into why real-world networks exhibit the above-mentioned properties. Moreover, a game-theoretic model allows measuring the impact of the selfish agent behavior on the whole society of agents. 

So far, game-theoretic approaches can explain the small-world property, that is, it has been proven that the diameter of all equilibrium networks is small~\cite{De07}. However, to the best of our knowledge, no known game-theoretic model can also explain the emergence of clustering and a power-law degree distribution. 

\paragraph{Our Contribution} In this paper, we propose and analyze a simple and very general game-theoretic model which is motivated by real-world social networks. Its main actors are selfish agents that bilaterally form costly links to increase their centrality. Hereby, the cost of each link is an arbitrary convex function in the distance of the involved nodes without this link. This naturally models the convention that connecting with a friend of a friend is much easier than to establish and maintain a link with an agent having no common acquaintances. To establish a link, both endpoints have to agree and pay the edge's cost.  

We characterize the social optimum and prove the existence of equilibrium networks for our generic model. For this, we provide many structural properties. Moreover, we give (nearly) tight bounds on the diameter, the Price of Anarchy and the Price of Stability that essentially only depend on the cost of closing a triangle and on the cost of maintaining a bridge-edge in the network. This implies that all these values are very low for many natural edge-cost functions. Moreover, as a proof of concept, we show via simulation experiments of our model that a given sparse initial network evolves over time into an equilibrium network having a power-law degree distribution, high clustering and a low diameter. Hence, our model promises to be the first game-theoretic network formation model which predicts networks that exhibit all core properties of real-world networks.  

\paragraph{Model and Notation}
We consider a model which is related to the bilateral network creation game~\cite{CP05}. In our model, called \emph{social network creation game (SNCG)}, the set of $n$ selfish agents $V$ corresponds to the nodes of a network and the agents' strategies determine the edge-set of the formed network $G$. More precisely, let $\mathbf{s} = (S_1,\dots,S_n)$ denote the strategy profile, where $S_u \subseteq V\setminus\{u\}$ corresponds to agent $u$'s strategy, then the jointly created network $G(\mathbf{s})$ is defined as $G(\mathbf{s}) = (V,E)$, with $E = \{\{u,v\}\mid u,v \in V, u\in S_v, v\in S_u\}$. And, inversely, for any given undirected network $G = (V,E)$ there exists a minimal strategy vector $\mathbf{s} = (S_1,\dots,S_{|V|})$ with $u \in S_v$ and $v\in S_u$ if and only if $\{u,v\} \in E$, that realizes this network, i.e., with $G = G(\mathbf{s})$. Hence, we will omit the reference to $\mathbf{s}$. Also, we will use the shorthand $uv$ for the undirected edge $\{u,v\}$.

Each agent $u$ tries to optimize a cost function $cost(G,u)$, which solely depends on the structure of the network $G$. In real-world social networks new connections are formed by a bilateral agreement of both endpoints while an existing connection can be unilaterally removed by any one of the involved endpoints. Following this idea, we consider only single edge additions with consent of both endpoints or single edge deletions as possible (joint) strategy changes of the agents. As equilibrium concept we adopt the well-known solution concept called \emph{pairwise stability}~\cite{JW96}. Intuitively, a network $G$ is \PS if every edge of $G$ is beneficial for both endpoints of the edge and for every non-edge of $G$, at least one endpoint of that edge would increase her cost by creating the edge. More formally, $G = (V,E)$ is \PS if and only if the following conditions hold:
\begin{enumerate}
\item for every edge $uv \in E$, we have $cost(G-uv,u) \geq cost(G,u)$ and $cost(G-uv,v) \geq cost(G,v)$;
\item for every non-edge $uv \notin E$, we have $cost(G+uv,u) \geq cost(G,u)$ or $cost(G+uv,v) \geq cost(G,v)$;
\end{enumerate}
where $G-uv$ (resp., $G+uv$) denotes the network $G$ in which the edge $uv$ has been deleted (resp., added).
Created edges are bidirectional and can be used by all agents, but the cost of  each edge is equally shared by its two endpoints.

We denote by $d_G(u,v)$ the distance between $u$ and $v$ in $G = (V,E)$, i.e., the number of edges in a shortest path between $u$ and $v$ in $G$. We assume that $d_G(u,v)=+\infty$ if no path between $u$ and $v$ exists in $G$.
The main novel feature of our model is the definition of the cost of any edge $uv \in E$, which is proportional to the distance of both endpoints without the respective edge, i.e., proportional to $d_{G-uv}(u,v)$. This is motivated by the fact that, in social networks, the probability of establishing a new connection between two parties is inversely proportional to their degree of separation. More precisely, let $\sigma:\mathbb{R} \rightarrow \mathbb{R}$ be a monotonically increasing convex function such that $\sigma(0)=0$.\footnote{All the results of this paper hold if we replace this constraint by the milder constraint $\sigma(3)\geq \frac{3}{2}\sigma(2)$.}
The cost of the edge $uv$ in network $G$ is equal to 
\begin{equation*}
c_G(uv)=	\begin{cases}
				\sigma\left(d_{G-uv}(u,v)\right) &\text{if $d_{G-uv}(u,v) \neq +\infty$,}\\
				\sigma(n) &\text{otherwise}.
			\end{cases}
\end{equation*} 
We call an edge $uv$ a \textit{$k$-edge} if  $d_{G-uv}(u,v)=k$, and a \textit{bridge} if $d_{G-uv}(u,v)=+\infty$. 
If the network is clear from the context, we will sometimes omit the reference to $G$ and we still simply write $c(uv)$ to denote the cost of edge $uv$. Note that by definition, any bridge in $G$, i.e., any edge whose removal would increase the number of connected components of $G$, has cost $\sigma(n)>\sigma(n-1)$ and thus any bridge has higher cost than any other non-bridge edge. 
The latter property can be understood as an incentive towards more robust networks. Note, that the addition or removal of an edge in network $G$ can also influence the cost of other edges in $G$. See Figure~\ref{fig:edge_cost_example} for an example.
\begin{figure}[h]
\centering
\includegraphics[width=9cm]{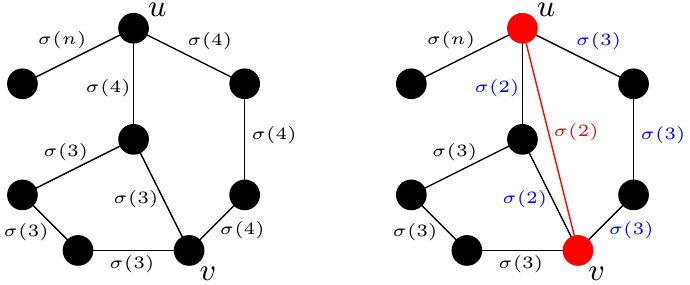}
\caption{Edge costs before and after the edge $uv$ is added.}
\label{fig:edge_cost_example}
\end{figure}

\noindent The resulting cost for an agent $u$ in the network $G$ is the sum of the cost of all edges incident to $u$ and the sum of distances to all other agents:
$$
cost(G, u):= \frac{1}{2}\sum_{v \in N_G(u)} c_{G}(uv) + \sum_{v\in V} d_{G}(u,v),
$$
where $N_G(u)$ is the set of all neighbors of $u$ in $G$. The quality of the created network is measured by its \textit{social cost}, which is denoted by $cost(G):=\sum_{u\in V}cost(G, u)$ and measures the total cost of all agents.

As in the original bilateral network creation game~\cite{CP05}, we restrict our study to connected networks, as any \PS non-connected network has an unbounded social cost\footnote{E.g., any network with no edges and $n\geq 3$ is \PS.}.  Let $worst_n$ (resp., $best_n$) be the highest (resp., lowest) possible social cost of a \PS (connected) network created by $n$ agents, assuming that a \PS network with $n$ agents always exists. 
Moreover, let $opt_n$ be the social cost of a {\em social optimum}, i.e., a minimum social cost network of $n$ nodes. 
Then the \textit{Price of Anarchy (PoA)}~\cite{KP99} is defined as $\max_{n\in\mathbb{N}}\frac{worst_n}{opt_n}$ 
and measures the deterioration of the network's social cost due to the agents' selfishness. 
The \textit{Price of Stability (PoS)}~\cite{ADKTWR} is the ratio $\max_{n\in\mathbb{N}}\frac{best_n}{opt_n}$ and describes the minimal cost discrepancy between an equilibrium and an optimal outcome.

\paragraph{Related Work} 
Strategic network formation is a rich and diverse research area and it is impossible to discuss all previous work in detail. Instead, we focus on the models which are closest to our approach. 

The SNCG is a variant of the bilateral network creation game~(BNCG)~\cite{CP05}. The BNCG is based on the unilateral network creation game~(NCG)~\cite{Fab03} where edges can be created and must be paid by only one of its endpoints, and the pure Nash equilibrium is used as solution concept. In both models edges have a uniform cost of $\alpha > 0$. In a series of papers it was established that the PoS for the NCG is at most $\frac{4}{3}$ and the PoA is constant for almost all values of $\alpha$~\cite{Fab03,Al06,De07,MS10,MMM13,BL20,AM19}. Moreover, it was shown that the diameter of any equilibrium network in the NCG is at most $2^{\mathcal{O}(\sqrt{\log n})}$~\cite{De07} and for many ranges of $\alpha$ it is constant. In contrast, for the BNCG it was shown the PoA of the BNCG is in $\Theta(\min\{\sqrt{\alpha},n/\sqrt{\alpha}\})$ and that a equilibrium networks with a diameter in $\Theta(\sqrt{\alpha})$ exist~\cite{CP05,De07}. The original NCG was dedicated to model real-world networks like peer-to-peer networks and social networks. However, the main downside of these classical models is that they do not predict a realistic degree distribution or high clustering. 
A NCG variant was proposed where agents try to maximize their local clustering instead of their centrality~\cite{BK11}. This model yields various sparse equilibrium networks with high clustering but these can have a large diameter and a homogeneous degree distribution.

Closer to the SNCG are variants of the NCG with non-uniform edge cost. Models were proposed where the edge cost is proportional to its quality~\cite{CMadH14}. Edges between certain types of agents have different costs~\cite{MMO14,MMO15}, and the edge cost depends on the node degree~\cite{CLMM17}, or the edge costs are defined by an underlying geometry~\cite{BFLM19}. Especially the latter is related to our model, as our model can also be understood as having a dynamically changing underlying geometry which depends on the structure of the current network. Finally, the island connection model~\cite{JR05} assumes that groups of agents are based on islands and that the edge cost within an island is much lower than across islands. This yields equilibria with low diameter and high clustering but no realistic degree distribution.

The SNCG incorporates a robustness aspect since bridge-edges are expensive. This fits to a recent trend in the AI community for studying robust network formation~\cite{MMO15,CLMM16,GJKKM16,CJKKM19,EFLM20}.

Despite the variety of studied network formation models, to the best of our knowledge, no simple game-theoretic model exists, which predicts a low diameter, a power-law degree distribution and high clustering in its equilibrium networks. We are also not aware of any simulation results in this direction. However, there are two promising but very specialized candidates in that direction. The first candidate, which is particularly tailored to the web graph~\cite{KMPRS15}, yields directed equilibrium networks that share many features of real-world content networks. The second candidate uses a game-theoretic framework and hyperbolic geometry to generate networks with real-world features. In the network navigation game~\cite{Gul15}, agents correspond to randomly sampled points in the hyperbolic plane and they strategically create edges to ensure greedy routing in the created network. It is shown that the equilibrium networks indeed have a power-law degree distribution and high clustering. However, the main reason for this is not the strategic behavior of the agents but the fact that the agents correspond to uniformly sampled points in the hyperbolic plane. It is known that the closely related hyperbolic random graphs~\cite{Kri10} indeed have all core properties of real-world networks. 

\section{Properties of Equilibrium Networks}
In this section we prove structural properties satisfied by all connected \PS networks that will be useful in proving our main results. We first define some basic notation and provide a nice property satisfied by the function $\sigma$. An edge $e$ of a network $G$ is a {\em bridge} if $G-e$ has at least one more connected component than $G$. A connected network that has no bridge is said to be {\em 2-edge-connected}. A {\em 2-edge-connected component} of a network $G$ is a maximal (w.r.t. node addition) induced subgraph of $G$ that is 2-edge-connected.\footnote{The subgraph of $G$ {\em induced} by a node set $U\subseteq V$ is a subgraph whose node set is $U$ that, for any two nodes $u,v \in U$, contains the edge $uv$ if $uv$ is also an edge of $G$.} The {\em diameter} $D$ of a network $G$ is equal to the length of the longest shortest path in $G$, i.e., $D=\max_{u,v \in V}d_G(u,v)$. Finally, we say that an edge $uv$ of $G$ is an {\em $i$-edge} if $d_{G-uv}(u,v)=i$, where we use the convention {\em $n$-edge} for a bridge edge.
\begin{proposition}\label{prop:basic_property_sigma}
Fix a positive real value $x$. Let $x_1, \dots , x_k$, with $0 \leq x_i \leq x$, be $k\geq 2$ positive real values and let $\lambda_1,\dots,\lambda_k$, with $\lambda \in [0,1]$, such that $x=\sum_{i=1}^k (\lambda_ix_i)$. Then $\sigma(x) \geq \sum_{i=1}^k\big(\lambda_i \sigma(x_i)\big)$.
\end{proposition}
\begin{proof}
We show that $\sigma(x_i) \leq \frac{x_i\sigma(x)}{x}$. This is enough to prove the claim since 
\[
\sum_{i=1}^k \big(\lambda_i \sigma(x_i)\big) \leq \frac{\sigma(x)}{x}\sum_{i=1}^k (\lambda_i x_i) \leq \sigma(x).
\] 
Let $x_i=\bar \lambda_i x$, i.e., $\bar \lambda_i=\frac{x_i}{x}$, and observe that $\bar \lambda_i \in [0,1]$. By convexity of $\sigma$ we have that
\[
\sigma(x_i)=\sigma\big((1-\bar \lambda_i) 0+\bar \lambda_i x\big) \leq (1-\bar \lambda_i)0+\bar \lambda_i \sigma(x)=\frac{x_i}{x} \sigma(x).
\]
\end{proof} 

\noindent In the next statement we claim that nodes can be incident to at most one expensive edge. Hence, the number of such edges is limited. 
\begin{proposition}\label{prop:edge_price_in_PSN}
In any \PS network, any node has at most one incident edge of cost at least $\sigma(4)$.
If $2\sigma(2)\leq \sigma(3)$ holds, any node in a \PS network has at most one incident edge of cost at least $\sigma(3)$.
\end{proposition}
\begin{proof}
Let $uv$ and $vw$ be two distinct edges of $G$ that are incident to $v$. We prove the claim by showing that at most one of these edges can have a cost of at least $2\sigma(2)$. This implies the claim since, by Proposition~\ref{prop:basic_property_sigma}, $\sigma(4) \geq 2\sigma(2)$.
 
If both edges $uv$ and $vw$ have a cost of at least $2\sigma(2)$ each, then $G$ is not \PS as, by adding the edge $uw$, the total edge cost of both agent $u$ and agent $w$ does not increase, while the total distance cost of each of the two agents decreases by at least 1. In fact, the edge cost of each edge $uv$, $uw$, and $vw$ in $G+uw$ is equal to $\sigma(2)$.
\end{proof}

\noindent Next, we establish that all \PS networks contain at most three bridges.

\begin{proposition}\label{prop:2_bridges_PSN}
Any \PS network contains at most three bridges.
\end{proposition}
\begin{proof}
We prove that any network $G$ with four or more bridges cannot be \PS. First, we show that there are two bridges at a distance of at most $n/2-2$ in $G$, i.e., two vertices $u_1$ and $u_2$ that are incident to 2 distinct bridges $e_1=u_1v_1$ and $e_2=u_2v_2$ such that $d_G(u_1,u_2) \leq \frac{n}{2}-2$. We observe that this is enough to prove the claim. Indeed, w.l.o.g., let $d_G(v_1,v_2)=2+d_G(u_1,u_2)$. We have that $d_G(v_1,v_2) \leq \frac{n}{2}$. By adding the edge $v_1v_2$ to $G$ the total distance cost of both agents $v_1$ and $v_2$ decreases by at least 1, while their total edge cost differs by $\sigma(n/2)-\frac{1}{2}\sigma(n) \leq 0$ as, by Proposition \ref{prop:basic_property_sigma}, $\sigma(n) \geq 2\sigma(n/2)$. Hence, $G$ cannot be \PS.

We now complete the proof by showing that there are two bridges at a distance of at most $\frac{n}{2}-2$.

Let $\mathcal{T}$ be a {\em block-cut tree} decomposition of $G$, i.e., a decomposition of $G$ into maximal {\em 2-connected components} and {\em cut nodes}.\footnote{A node $x$ of a connected graph $G$ is a {\em cut node} if its removal from $G$ results in a graph that is not connected. A {\em 2-connected graph} is a connected graph with no cut node. A {\em 2-connected component} of $G$ is a maximal (w.r.t. node insertion) 2-connected subgraph of $G$. A {\em block-cut tree} $\mathcal{T}$ of $G$ is a tree where each tree node represents either a cut node or a 2-connected component of $G$. More precisely, there is an edge between the representative of a cut node $x$ of $G$ and the representative of a 2-connected component $G'$ of $G$ iff $x$ is a node of $G'$.} Notice that each bridge $uv$ is represented in $\mathcal{T}$ as a 2-connected component that is connected with the 2 cut nodes $u$ and $v$. Let $\mathcal{T}'$ be a minimal connected subtree of $\mathcal{T}$ that contains exactly four 2-connected components that are bridges, and let $G'$ be the subgraph of $G$ whose block-cut tree decomposition is represented by $\mathcal{T}'$. Let $e_i=u_i v_i$, with $i=1,2,3$, be the 4 bridges of $G'$. We denote by $\dist(e_i,e_j)=\min\big\{d_{G'}(u_i,u_j),d_{G'}(u_i,v_j),d_{G'}(v_i,u_j),d_{G'}(v_i,v_j)\big\}$ the distance in $G'$ between the two bridges $e_i$ and $e_j$. We shall prove that $\min_{1\leq i < j \leq 4}\dist(e_i,e_j)$.

Let $P_{i,j}$ denote the (unique) simple path in $\mathcal{T}'$ from the node that corresponds to the 2-connected component corresponding to the node that corresponds to the 2-connected component $e_j$. The proof divides into 2 complementary cases, depending on the structures of the paths $P_{i,j}$, with $1 \leq i < j \leq 4$. 

The first case is when at least two paths in $\{P_{i,j} \mid 1 \leq i < j \leq 4\}$ are node disjoint. W.l.o.g., we assume that $P_{1,2}$ is node disjoint w.r.t. $P_{3,4}$. W.l.o.g., we assume that the overall number of nodes of the 2-connected components corresponding to the nodes in $P_{1,2}$ is at most $n/2$. It is well-known that a 2-edge-connected graph with $n$ nodes has diameter of at most $\frac{2}{3}n$~\cite{Caccetta92}. In this case, as $n > 4$ ($G$ has at least four bridges), we have that
\[
\dist(e_1,e_2) \leq \frac{2}{3}\left(\frac{n}{2}-2\right) \leq \frac{n}{3}-\frac{4}{3} \leq \frac{n}{2}-2.
\]

The second case is when there are no two node disjoint paths in $\{P_{i,j} \mid 1 \leq i < j \leq 4\}$. This can happen only if there is exactly one 2-connected component, say $C$ that is traversed by all the 4 paths $P_{i,j}$'s. Let $n_C$ denote the number of nodes of $C$. Let $n_i$, with $i=1,2,3,4$, denote the overall number of nodes of the 2-connected components corresponding to the nodes in the (unique) simple path in $\mathcal{T}'$ from the node that represents $e_i$ to the node that represents the 2-connected components right before $C$. Clearly, $n_1,n_2,n_3,n_4 \geq 2$ and $n_C\leq n-4$. W.l.o.g., we assume that $n_1 \leq n_2 \leq n_3 \leq n_4$. We prove that $\dist(e_1,e_2) \leq \frac{n}{2}-2$. 
It is well known that the diameter of any 2-connected graph with $n$ nodes is at most $\left \lceil \frac{n-1}{2}\right \rceil$~\cite{Caccetta92}. We divide the proof into 3 subcases.
The first subcase is when $n_1,n_2=2$. This implies that 
\[
\dist(e_1,e_2) \leq \left\lceil\frac{n_C-1}{2}\right\rceil \leq \left \lceil \frac{n-4-1}{2}\right\rceil \leq \frac{n}{2}-2.
\]
The second subcase is when $n_1=2$ and $n_2 > 2$. In this case, we have that $n_2,n_3,n_4 \geq 4$. Moreover, $n_C+n_1+n_2+n_3+n_4 \leq n+4$ from which we derive $n_C\leq n+2-3n_2$. The diameter of $C$ is at most $\lceil (n_C-1)/2 \rceil$ as $C$ is 2-connected. Moreover, all the other 2-connected components traversed by $P_{1,2}$, except for $e_1$ and $e_2$ and $C$, form a 2-edge-connected graphs of diameter at most $\frac{2}{3}(n_2-1)$. Therefore, we have that
\begin{align*}
\dist(e_1,e_2)	& \leq \left\lceil \frac{n_C-1}{2}\right\rceil + \frac{2}{3}(n_2-1) \leq \frac{n_C}{2}+\frac{2}{3}n_2-\frac{2}{3}\\
				& \leq \frac{n}{2}+1-\frac{3}{2}n_2+\frac{2}{3}n_2-1 \leq \frac{n}{2}-\frac{5}{6}n_2 < \frac{n}{2}-2.
\end{align*}
The third and last subcase is when $n_1,n_2 \geq 2$. In this case, we have that $n_1,n_2,n_3,n_4 \geq 4$. Moreover, $n_C+n_1+n_2+n_3+n_4 \leq n+4$ from which we derive $n_C\leq n+4-\frac{4}{3}(n_1+n_2)-\frac{2\cdot 8}{3}= n-\frac{4}{3}(n_1+n_2)-\frac{4}{3}$.
 The diameter of $C$ is at most $\lceil (n_C-1)/2 \rceil$ as $C$ is 2-connected. Moreover, all the other 2-connected components traversed by $P_{1,2}$, except for $e_1$ and $e_2$ and $C$, form two 2-edge-connected graphs of diameter at most $\frac{2}{3}(n_1-1)$ and $\frac{2}{3}(n_2-1)$, respectively. Therefore, we have that
\begin{align*}
\dist(e_1,e_2)	& \leq \left\lceil \frac{n_C-1}{2}\right\rceil + \frac{2}{3}(n_1-1)+\frac{2}{3}(n_2-1)\\
				& \leq \frac{n_C}{2}+\frac{2}{3}(n_1+n_2)-\frac{4}{3}\\
				& \leq \frac{n}{2}-\frac{2}{3}(n_1+n_2)-\frac{2}{3}+\frac{2}{3}(n_1+n_2)-\frac{4}{3}\\
				& = \frac{n}{2}-2.
\end{align*}
This completes the proof.
\end{proof}

\noindent The following proposition shows an upper bound of the diameter of any \PS network that only depends on the cost of edges which close a triangle.
\begin{proposition}\label{prop:diam_PSN}
The diameter of any \PS network is at most $\sigma(2) + 2$.
\end{proposition}
\begin{proof}
Consider a \PS network $G$ of diameter~$D$. 
Let $v_0, v_1,\ldots, v_{D}$ be a diametral path of $G$.
Consider the addition of the edge between $v_{\lfloor D/2\rfloor -1}$ and $v_{\lfloor D/2\rfloor+1}$ to network $G$. Each node $v_0,\dots,v_{\lfloor D/2\rfloor -1}$ becomes 1 unit closer to $v_{\lfloor D/2\rfloor+1}$; similarly, each node $v_{\lfloor D/2\rfloor+1}, \dots, v_D$ becomes 1 unit closer to $v_{\lfloor D/2\rfloor-1}$. In both cases, the distance cost of the considered agent decreases by at least $\lfloor D/2\rfloor$. Since the network is \PS, both agents $v_{\lfloor D/2\rfloor-1}$ and $v_{\lfloor D/2\rfloor+1}$ have no incentive in buying the considered edge. Therefore, $\sigma(2)/2 - \lfloor D/2\rfloor \geq 0$ from which we derive $D\leq \sigma(2) + 2$. \qedhere
\end{proof}

\noindent Finally, we prove an upper bound on the cost of non-bridge edges. This implies that all \PS networks contain only small minimal cycles, i.e., cycles where the shortest path between two nodes in the cycle is along the cycle.

\begin{proposition}\label{prop:max_cost_of_expensive_edges_PS}
In a \PS network, for all $k\notin\{2, 3, n\}$, the cost of any $k$-edge is $\sigma(k) < n\sigma(2)$.
If $\sigma(2)\leq \frac{1}{2}\sigma(3)$ holds, for all $k\notin\{2, n\}$, the cost of any $k$-edge is $\sigma(k) \leq n\sigma(2)$.
\end{proposition}
\begin{proof}
Consider a \PS network $G$. 
Assume to the contrary that there is a non-bridge $k$-edge $uv$ in $G$ of cost $\sigma(k)$. 
Consider the deletion of the edge $uv$ by one of its endpoints, say $u$. 
Let $V_v$ (resp., $V_u$) be a subset of nodes such that all shortest paths between $u$ (resp., $v$) and any node in $V_v$ (resp., $V_u$) goes through the edge $uv$.  
By Proposition~\ref{prop:edge_price_in_PSN}, all other edges incident to either $u$ or $v$ are 2-edges and 3-edges. 
As a consequence, the deletion of the edge $uv$ does not increase the cost of the edges incident to $u$ and $v$. Therefore, the edge cost of $u$ (resp., $v$) decreases by $\frac{1}{2}c_G(uv)$, while the distance cost of $u$ increases by at most $|V_v|(\dist_{G-uv}(u,v)-1)=|V_v|(k-1)$ (resp., $|V_u|(k-1)$). 
Since $G$ is \PS, the $u$'s (resp., $v$) cost difference is greater than zero, i.e., $-\sigma(k)/2+|V_v|(k-1)\geq 0$ (resp.,  $-\sigma(k)/2+|V_u|\cdot(k-1)\geq 0$). 
We sum up the two inequalities and get $\sigma(k)\leq (|V_v|+|V_u|)\cdot(k-1)$. 
Note that $V_v\cap V_u=\emptyset$. 
Indeed, if there is $x\in V_v\cap V_u$, then $\dist_{G}(u,x)=1+\dist_{G}(v,x)=1+1+d_G(u,x)$, i.e., $0=2$. 
Therefore, $\sigma(k)\leq n(k-1)$.

If we assume $\sigma(2)\leq \frac{1}{2}\sigma(3)$, each node has at most one incident 3-edge according to Proposition~\ref{prop:edge_price_in_PSN}. 
Similarly to the above proof, we obtain $\sigma(3)\leq 2n$. 

Now we consider the addition of a 2-edge $u'v$ in the minimal cycle of length $(k+1)$ that contains the edge $uv$ by the node $v$ and a neighbor $u'$ of $u$ in the cycle.  
First, we consider $k\geq 4$. 
For both endpoints, this move improves the distance to at least $\frac{k+1-3}{2}$ nodes in the cycle. 
Moreover, by  Proposition~\ref{prop:edge_price_in_PSN}, all other of $v$'s incident edges are 2- or 3-edges; therefore $v$ has at least one neighbor $v'$ that is not in the cycle and $\dist_G(u', v')\geq k-1\geq 3$ (otherwise, it would not be a $k$-edge). 
Analogously, $u'$ has a neighbor outside of the cycle at distance at least 3 from $v$. 
This implies that both endpoints of the edge $u'v$ will improve their distance to at least $\frac{k-2}{2}+1$ nodes by 1 after adding the edge.   
Since $G$ is \PS, this move is not profitable, i.e., $\frac{\sigma(2)}{2} - \frac{k-2}{2}-1\geq 0$. 
Hence, $\sigma(2)\geq k$. 
 Combining this inequality with the inequality $\sigma(k)\leq n(k-1)$ from the first part of the proof, we get $\sigma(k)< n\sigma(2)$, if $k\geq 4$.
 If $k=3$ and $\sigma(2)\leq \frac{1}{2}\sigma(3)$, the addition of a 2-edge can improve the distance to only one node. 
 Since we assume that $G$ is \PS, $\frac{1}{2}\sigma(2)-1\geq 0$, i.e., $\sigma(2)\geq 2$. 
 Combining this inequality with the above inequality for 3-edges, we get $\sigma(3)\leq 2n\leq n\sigma(2)$. 
 The statement follows.\qedhere
\end{proof}

\section{Equilibrium Existence and Social Optima}
The \textit{clique graph} of $n$ nodes is denoted by $\mathbf{K}_n$.
A \textit{fan graph} $\mathbf{F}_n$ with $n$ nodes consists of a star with $n-1$ leaves $v_0,\dots,v_{n-2}$ augmented with all the edges of the form $v_{2i}v_{2i+1}$, for $i=0,\dots, \lfloor \frac{n-2}{2}\rfloor$, where all indices are computed modulo $n-2$ (see Figure~\ref{fig:fan_graphs} for examples). In other words, $\mathbf{F}_n$, with $n$ odd, is a star augmented with a perfect matching w.r.t. the star leaves,\footnote{In the literature, this graph is also known as {\em friendship graph} or {\em Dutch windmill graph}.} while  $\mathbf{F}_n$, with $n$ even, consists of $\mathbf{F}_{n-1}$ augmented with an additional node that is connected to the star center and any star leaf. 
\begin{figure}[h]
\centering
\includegraphics[width=7cm]{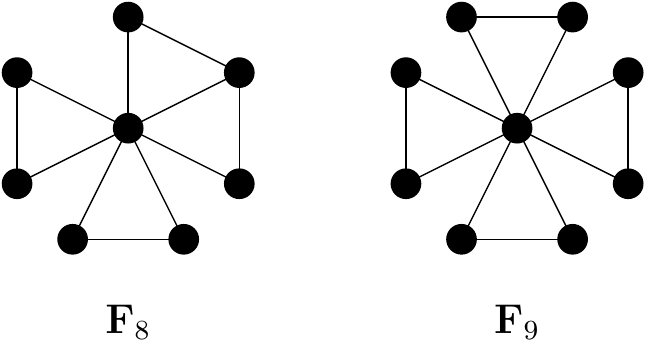}
\caption{Two examples of fan graphs.}
\label{fig:fan_graphs}
\end{figure}
Clique graphs and fan graphs play an important role since, as we will prove, the former are social optima when $\sigma(2) \leq 2$, while the latter are social optima when $\sigma(2)\geq 2$. Furthermore, we also show that clique graphs are \PS whenever $\sigma(2) \leq 2$, while (almost) fan graphs are \PS whenever $\sigma(2)\geq 2$.

\begin{theorem}\label{thm:fan_and_clique_are_OPT}
If $\sigma(2) < 2$, then $\mathbf{K}_n$ is the unique social optimum. If $\sigma(2) > 2$, then $\mathbf{F}_n$ is the unique social optimum. Finally, if $\sigma(2)=2$ any network of diameter 2 and containing only $2$-edges is a social optimum.\footnote{Hence, for $\sigma(2)=2$, $\mathbf{K}_n$ and $\mathbf{F}_n$ are also social optima.}
\end{theorem} 
\begin{proof}
Let $G$ be a social optimum and assume that $G$ contains $m$ edges. We provide a lower bound $LB(m)$ to the social cost of $G$. The edge cost of $G$ is lower bounded by $m \sigma(2)$; moreover, such a lower bound is matched only by networks whose edges are all 2-edges. The distance cost of $G$ is lower bounded by $2(n(n-1)-2m)+2m=2n(n-1)-2m$ as there are $2m$ ordered pairs of nodes that are at distance 1 (two pairs for each of the $m$ edges) and $n(n-1)-2m$ ordered pairs of nodes that are at distance greater than or equal to 2. We observe that the lower bound on the distance cost is matched only by networks of diameter~2. Therefore, the social cost is lower bounded by 
\[LB(m):=2n(n-1)+(\sigma(2)-2)m\]
and the value $LB(m)$ can be matched only by networks of diameter 2 whose edges are all 2-edges.

Clearly, when $\sigma(2) < 2$, we have that $LB(m)$ is minimized when $m$ is maximized; furthermore, $\mathbf{K}_n$ is the only network whose social cost matches the lower bound. Hence, for $\sigma(2) < 2$, $\mathbf{K}_n$ is the unique social optimum. For the case $\sigma(2)=2$, we have that the value $LB(m)=2n(n-1)$ is matched by any network of diameter $2$ that contains only $2$-edges. Hence, all and only such networks are social optima.

Now we prove the theorem statement for the remaining case in which $\sigma(2) > 2$. We have to show that $\mathbf{F}_n$ is the unique social optimum. First of all we observe that $LB(m)$ is minimized when $m$ is minimized. Since the social cost of $\mathbf{F}_n$ matches the lower bound $LB(m')$ with $m'=n-1+\left\lceil\frac{n-1}{2}\right\rceil$, we have that $m \leq n-1+\lceil \frac{n-1}{2}\rceil$. Moreover, the social cost of the network is bounded only if the network is connected. This implies that $G$ is connected, which in turn implies that $m \geq n-1$. Therefore, to prove that $\mathbf{F}_n$ is the unique social optimum, it is enough to prove that any network with $m$ edges, with $n-1 \leq m \leq n-1+\left\lceil \frac{n-1}{2}\right\rceil$, has a social cost that is strictly larger than $2n(n-1)+(\sigma(2)-2)\left((n-1)+\left\lceil \frac{n-1}{2} \right\rceil\right)$, unless it is  isomorphic to $\mathbf{F}_n$.

For the rest of the proof we can also assume that $n \geq 4$. In fact, for $n=2$ it is clear that $\mathbf{F}_2=\mathbf{K}_2$ is the social optimum since this is the only connected network of $2$ nodes. Moreover, for $n=3$ we have that $\mathbf{F}_3=\mathbf{K}_3$ is again the only social optimum. Indeed, $\mathbf{F}_3$ has a social cost of $3\sigma(2)+6$, while the unique other connected network -- i.e., the path of length 2 -- has a social cost of $2\sigma(3)+8$. Since by Proposition~\ref{prop:basic_property_sigma}, $\sigma(3)\geq \frac{3}{2}\sigma(2)$ (indeed $3\geq 1\cdot 2+0.5\cdot 2$), it follows that $2\sigma(3)+8 \geq 3\sigma(2)+8 > 3\sigma(2)+6$.

Moreover, we can also assume that $G$ is 2-edge-connected. Indeed, let $G$ be a network that contains a bridge, say $uv$. W.l.o.g., let $u' \neq v$ be a neighbor of $u$, whose existence is guaranteed since $n\geq 4$. The distance cost of $G+u'v$ is strictly smaller than the distance cost of $G$. Moreover, the edge cost of $G+u'v$ is at most the edge cost of $G$. In fact, the cost of the bridge $uv$ in $G$ is at least $\sigma(4)$, while the cost of the two edges $uv$ and $u'v$ in $G+u'v$ is at most $2\sigma(2)$ and, by Proposition~\ref{prop:basic_property_sigma}, we have that $\sigma(4) \geq 2\sigma(2)$. As a consequence, the social cost of $G+u'v$ is strictly smaller than the social cost of $G$.

We divide the proof into two cases, depending on whether $m < n-1+\left\lceil \frac{n-1}{2}\right\rceil$ or not. 

We consider the case in which  $m < n-1+\left\lceil \frac{n-1}{2}\right\rceil$. 
Consider the subgraph $H$ of $G$ that is induced by 2-edges only. Such a subgraph contains $k \geq 1$ connected components, $h$ of which are singleton nodes. Let $C_1,\dots,C_{k-h}$ be the non-trivial connected components of $H$. Each $C_i$ contains $n_i \geq 3$ nodes and $m_i$ edges, where $m_i \geq n_i-1 + \left \lceil\frac{n_i-1}{2}\right\rceil$. Indeed, each $C_i$ can be generated starting from a triangle, i.e., $\mathbf{K}_3$, and by iteratively adding either one or two nodes so as the resulting induced subgraph contains at least one more triangle than before. 

Clearly, at each step, we add either one node and at least two edges or two nodes and at least three edges. Obviously, the number of edges is minimized when we add two new nodes and exactly three edges at each step. Therefore, when $n_i$ is odd, i.e., $n_i-3$ is even, we add at least three edges for every two nodes; when $n_i$ is even, i.e., $n_i-3$ is odd, we add at least three edges for every two nodes except one node and at least two edges for the remaining node. As a consequence, when $n_i$ is odd, we have $m_i \geq 3+3\frac{n_i-3}{2} = n_i-1+\left\lceil \frac{n_i-1}{2}\right\rceil$; when $n_i$ is even, we have $m_i \geq 3+3\frac{n_i-4}{2}+2= n_i-1+\left\lceil \frac{n_i-1}{2}\right\rceil$. In either case, $m_i \geq n_i-1+\left\lceil \frac{n_i-1}{2}\right\rceil$.

First of all, we observe that $n=h+\sum_{i=1}^{k-h}n_i$. Furthermore, since we are assuming $m < n-1+\left\lceil \frac{n-1}{2}\right\rceil$ it must be the case that $k \geq 2$. Indeed, for $k=1$, $h$ would be equal to $0$ and therefore $m_1 \geq n-1+\left\lceil \frac{n-1}{2}\right\rceil$. Finally, since we are assuming that $G$ is 2-edge-connected, there are at least $k$ edges of $G$ each of which connects a node of one connected component with a node of another connected component. Clearly, the cost of each of such $k$ edges is greater than or equal to $\sigma(3)$ each. Therefore, since by Proposition~\ref{prop:basic_property_sigma} $\sigma(3) \geq \frac{3}{2}\sigma(2)$, the overall edge cost of the network $G$ is lower bounded by
\begin{align*}
k\sigma(3)+\sigma(2)\sum_{i=1}^{k-h}m_i	& \geq \frac{3}{2}k\sigma(2)+\sigma(2)\sum_{i=1}^{k-h}\left(\frac{3}{2}(n_i-1)\right)\\
										& = \frac{3}{2}k\sigma(2)+ \frac{3}{2}(n-k)\sigma(2)\\
										& > \left(n-1+\left\lceil \frac{n-1}{2}\right\rceil\right)\sigma(2).
\end{align*}
As the distance cost of $G$ is lower bounded by $2n(n-1)-2m$, the overall social cost of $G$ is strictly larger than $2n(n-1)+(\sigma(2)-2)\left((n-1)+\left\lceil \frac{n-1}{2} \right\rceil\right)$, i.e., the social cost of $\mathbf{F}_n$. Therefore, no network with $m < n-1+\left\lceil \frac{n-1}{2}\right\rceil$ can be a social optimum.

We now consider the remaining case in which $m = n-1+\left\lceil \frac{n-1}{2}\right\rceil$ and show that $G$ is isomorphic to $\mathbf{F}_n$. First of all we observe that the social cost of $G$ cannot be smaller than the social cost of $\mathbf{F}_n$ as the social cost of $\mathbf{F}_n$ matches the value $LB(m)$. This implies that $\mathbf{F}_n$ is a social optimum. For the sake of contradiction, assume that $G$ is not isomorphic to $\mathbf{F}_n$. We show that $G$ must satisfy some structural properties, based on three important observations. The first observation is that $G$ consists of only 2-edges and has diameter 2, as otherwise the social cost of $G$ would be strictly larger than the value $LB(m)$.
The second observation is that $G$ has minimum degree equal to 2. Indeed, any network of minimum degree greater than or equal to 3 would have at least $m \geq \frac{3}{2} n > n-1+\left\lceil \frac{n-1}{2}\right\rceil$ edges. The third and last observation is that $G$ cannot have a node of degree $n-1$. As a consequence of these three observations, $G$ has the following structure:
\begin{itemize}
\item it contains a node $v$ that is connected to only two nodes, say $u$ and $u'$;
\item it contains the edge $uu'$ (otherwise the edges $uv$ and $u'v$ would not be 2-edges);
\item $u$ and $u'$ form a dominating set (otherwise the diameter of $G$ would be at least 3);
\item there is at least one node that is connected to $u$ but not to $u'$ and one node that is connected to $u'$ but not to $u$ (otherwise $G$ would contain a spanning star).
\end{itemize}
Let $A$ be the set of nodes different from $u'$ that are connected to $u$ but not to $u'$; similarly, let $B$ be the set of nodes different from $u$ that are connected to $u'$ but not to $u$. Finally, let $C$ be the set of nodes different from $v$ that are connected to both $u$ and $u'$ (see Figure~\ref{fig:Figure_thm1}).
\begin{figure}[h]
\centering
\includegraphics[width=6cm]{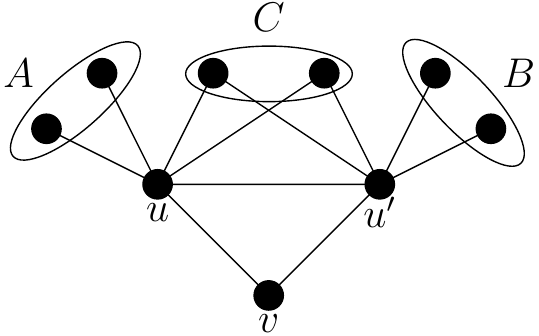}
\caption{The structure of a 2-edge-connected network of diameter 2 with minimum degree 2 that is not isomorphic to $\mathbf{F}_n$. The node $v$ has degree 2, the sets $A$ and $B$ are not empty, while $C$ can be empty. We have $|C| \geq 1$ if no edge $ab$ with $a \in A$ and $b \in B$ exists. All the edges of the network are 2-edges only if each node $a \in A$ (resp., $b \in B$) is connected either with a node of $C$ or another node of $A$ (resp., $B$).}
\label{fig:Figure_thm1}
\end{figure}
We observe that $|A|+|B|+|C|=n-3$. Each node $a$ of $A$ must be connected to another node of $A \cup C$ as otherwise the edge $ua$ would not be a 2-edge; similarly, each node $b$ of $B$ must be connected to another node of $B \cup C$ as otherwise the edge $u'b$ would not be a 2-edge. Furthermore, since $G$ has diameter 2, it must be the case that for any two nodes $a \in A$ and $b \in B$ either $ab$ is an edge of $G$ or there is node $z \in C$ such that $az$ and $zb$ are both edges of $G$.

If we assume the existence of an edge $ab$ for $a \in A$ and $b \in B$, then we can lower bound the number of edges of $G$ with 
$
m \geq 1+n+|C|+\left\lceil \frac{|A|}{2}\right\rceil+\left\lceil \frac{|B|}{2}\right\rceil \geq n+1 + \left\lceil \frac{n-3}{2}\right\rceil  > n-1 + \left\lceil \frac{n-1}{2}\right\rceil = m.
$

If we assume that no edge $ab$ with $a \in A$ and $b \in B$ exists, then $|C| \geq 1$ and we can lower bound the number of edges of $G$ with $m \geq n + |C| + |A| + |B| = n+(n-3) > n-1 + \left\lceil \frac{n-1}{2}\right\rceil=m$, as $n\geq 3+|A|+|B|+|C| > 6$. 

In either case, we have obtained a contradiction. Hence, $\mathbf{F}_n$ is the unique social optimum when $\sigma(2) > 2$.
\end{proof}

\noindent Now we prove the existence of \PS networks. For this we consider a modified fan graph $\mathbf{F}'_n$ that is equal to $\mathbf{F}_n$ if $n$ is odd. If $n$ is even, $\mathbf{F}'_n$ consists of  $\mathbf{F}_{n-1}$ and one additional node connected to the center. 

\begin{theorem}\label{thm:fan_and_clique_are_PS}
For $\sigma(2)\geq 2$, a the modified fan graph $\mathbf{F}'_n$ is a \PS network, otherwise a clique is the unique \PS network.
\end{theorem}
\begin{proof}
First we show that if $\sigma(2)< 2$, any \PS network is a clique. 
Assume to the contrary that there is a network $G$ that is \PS but at least one edge is missing. 
Adding any missing 2-edge costs $\sigma(2)/2$ for both its endpoints and improves the distance cost by at least 1. 
Thus, if $\sigma(2) < 2$, this move is an improvement.

Next, we prove that the modified fan graph is \PS for  $\sigma(2)\geq 2$. 
Assume $n\geq 3$, otherwise $\mathbf{F}_n'$ is trivially stable.
If $n$ is odd, deletion of any edge from $\mathbf{F}_n'$ increases the edge cost by $(\sigma(n)-2\sigma(2))/2 > 0$ and increases distance between its endpoints by 1. 
Thus, any edge removal is not an improvement. 
On the other hand, buying any edge which is not present in $\mathbf{F}_n'$ costs $\sigma(2)/2$ and improves the distance only between its endpoints, i.e., it improves the distance cost by 1. 
Hence, since $\sigma(2)\geq 2$, an addition of any extra edge to the modified fan graph is not an improvement. 

If $n$ is even, then the pairwise stability of $\mathbf{F}_n'$ follows from an analogous proof as for odd $n$ and from the observation that creating an edge with a leaf is not profitable.
\end{proof}

\section{PoA and PoS}
Here we prove upper and lower bounds to the PoA and PoS.
\begin{theorem}\label{thm:UB_PoA_PSN}
The PoA of the SNCG is in $O\left(\min\{\sigma(2),n\}+\frac{\sigma(n)}{n\max\{\sigma(2),n\}}\right)$. For the class of 2-edge-connected networks the PoA is in $O\big(\min\{\sigma(2),n\}\big)$.
\end{theorem}
\begin{proof}
By Theorem~\ref{thm:fan_and_clique_are_OPT} and Theorem~\ref{thm:fan_and_clique_are_PS}, we only need to focus on the case $\sigma(2) \geq 2$. Indeed, when $\sigma(2) < 2$, $\mathbf{K}_n$ is the unique \PS network as well as the unique social optimum and, therefore, the PoA is equal to 1.

For the rest of the proof we assume that $\sigma(2) \geq 2$. By Theorem~\ref{thm:fan_and_clique_are_OPT}, $\mathbf{F}_n$ is a social optimum of cost $\Omega\big(n^2+\sigma(2)n\big)=\Omega\big(n\max\{\sigma(2),n\}\big)$.
Consider a \PS network $G$ of maximum social cost for a given number of nodes $n$. Let $D$ be the diameter of $G$.
A trivial upper bound for the distance cost of the network is $n(n-1)\cdot D$. 
By Proposition 4, the network diameter is at most $\sigma(2)+2$, hence the distance cost of $G$ is at most $(\sigma(2)+2)\cdot n(n-1)$.

Now we will show an upper bound for the edge cost. 
Let $k_i$ denote the number of $i$-edges in $G$. 
By Proposition~\ref{prop:2_bridges_PSN}, $G$ has at most 3 bridges. 
Hence, for any \PS network we have that $k_n\leq 3$; if the network is additionally 2-edge-connected, then $k_n=0$. 
We consider two cases, depending on whether $2\sigma(2)\leq \sigma(3)$ or not.

We consider the case $2\sigma(2)\leq \sigma(3)$. By Proposition~\ref{prop:edge_price_in_PSN}, each node has at most one incident $i$-edge where $3\leq i < n$.  
Moreover, by Proposition~\ref{prop:max_cost_of_expensive_edges_PS}, $\sigma(i)\leq n\sigma(2)$ for any $i\geq 3$. 
Then the overall edge-cost of the network is at most 
\begin{align*}
&k_2\cdot \sigma(2) + \sum\limits_{i=3}^{n-1}\left(\sigma(i)\cdot k_i\right) + k_n\sigma(n)\\
&\leq \left(\frac{n(n-1)}{2}-\sum\limits_{i=3}^{n-1} k_i\right)\cdot\sigma(2) + n\sigma(2)\sum\limits_{i=3}^{n-1} k_i  + k_n\sigma(n)\\
&\leq \sigma(2)\cdot\frac{n(n-1)}{2} + (n-1)\sigma(2)\cdot \frac{n}{2} + k_n\sigma(n)\\
& \leq \sigma(2) n^2 + k_n\sigma(n).
\end{align*}
As a consequence, the PoA is at most in $\frac{2(\sigma(2)+1)n^2 + k_n\sigma(n)}{\Omega\big(n\max\{\sigma(2),n\}\big)}$.
It implies that the PoA is in $O\left(\min\{\sigma(2),n\}+\frac{\sigma(n)}{n\max\{\sigma(2),n\}}\right)$, while for the class of 2-edge-connected networks, the PoA is in $O\big(\min\{\sigma(2),n\}\big)$.

If $2\sigma(2) > \sigma(3)$, the upper bound for the edge cost from Proposition~\ref{prop:max_cost_of_expensive_edges_PS} holds for all $i$-edges such that $4\leq i\leq n-1$. 
To estimate the number of 3-edges, we use result from extremal graph theory. Mantel's Theorem~\cite{Mantel07} we have $k_3\leq \frac{n^2}{4}$. 
Then the edge cost of $G$ is at most 
\begin{align*}
&k_2\cdot \sigma(2) + k_3\cdot\sigma(3) + \sum\limits_{i=4}^{n-1}\left(\sigma(i)\cdot k_i\right) + k_n\sigma(n)\\
&\leq \left(\frac{n^2}{2}-\sum\limits_{i=3}^{n-1} k_i\right)\cdot\sigma(2) + \frac{n^2}{4}\sigma(3) +  n\sigma(2)\frac{n}{2}  + k_n\sigma(n)\\
&< \sigma(2)\cdot n^2 +\frac{n^2}{4}\cdot\sigma(2) + k_n\sigma(n)\\
&= \frac{5}{4}\sigma(2)\cdot n^2 + k_n\sigma(n).
\end{align*}
As in the previous case, we get that the PoA is in $O\left(\min\{\sigma(2),n\}+\frac{\sigma(n)}{n\max\{\sigma(2),n\}}\right)$, while for the class of 2-edge-connected networks, the PoA is $O\big(\min\{\sigma(2),n\}\big)$.
\end{proof}

\noindent It is worth noticing that the high inefficiency of worst case \PS networks in Theorem~\ref{thm:UB_PoA_PSN} follows from the existence of bridges in a network. The PoA is much better in bridge-free \PS networks. Such networks can for example evolve via edge additions starting from a 2-edge-connected network.  A real-world example for this would be co-authorship networks of authors with at least two papers.

We now prove lower bounds on the PoA. We start with the construction of a \PS 2-edge-connected network with high social cost and a diameter in $\Omega(\sigma(2))$.

\begin{lemma}\label{lemma:pairwise_stable_network_large_diameter}
There are 2-edge-connected \PS networks with $n=\Omega(\sigma(2))$ nodes, social cost in $\Omega\big(\sigma(2)n^2\big)$, and diameter of at least $\frac{\sigma(2)}{4}$.
\end{lemma}
\begin{proof}
Let $k \geq 2$ be an integer and $n=2\left\lceil \frac{\sigma(2)}{8} \right\rceil k+1$. The \PS network $G$ of $n$ nodes is obtained from a spider graph $\mathcal{S}$ with center $x$, with $\left\lceil \frac{\sigma(2)}{8} \right\rceil k+1$ nodes, and having $k$ legs of length $\left\lceil \frac{\sigma(2)}{8} \right\rceil$ each, where we add a new node for each edge of $\mathcal{S}$ that is connected to both endpoints of the respective edge (see Figure~\ref{fig:Lower_bound_PoA} for an example). 
\begin{figure}[h]
\centering
\includegraphics[width=11cm]{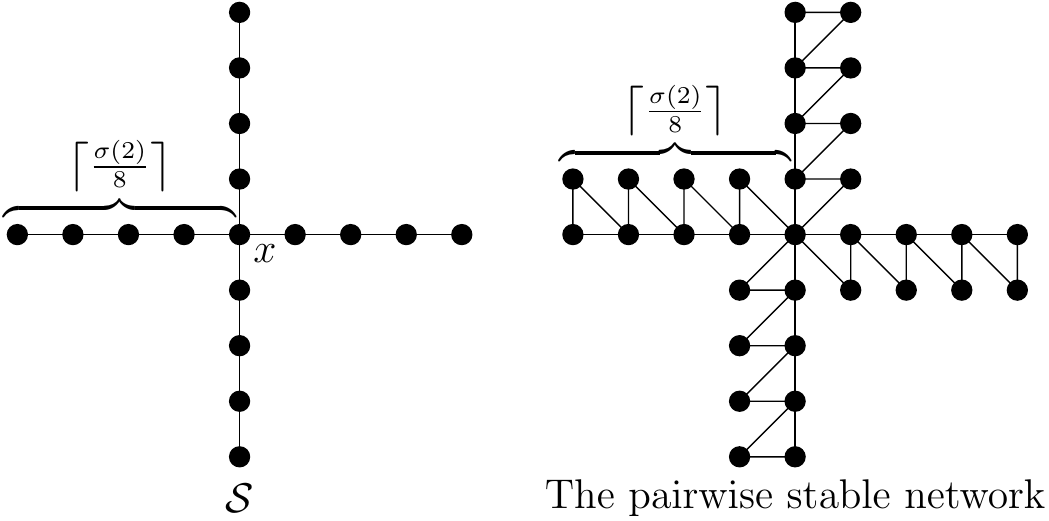}
\caption{Illustration of the construction in Lemma~\ref{lemma:pairwise_stable_network_large_diameter} for $k=4$ and $25 \leq \sigma(2) \leq 32$. Left: the spider $\mathcal{S}$ with $k$ legs of length $\left\lceil \frac{\sigma(2)}{8} \right\rceil=4$ each. Right: the \PS network obtained by augmenting $\mathcal{S}$ with $\left\lceil \frac{\sigma(2)}{8} \right\rceil k$ new nodes.}
\label{fig:Lower_bound_PoA}
\end{figure}
The distance cost of $G$ is at least $n$ times the distance cost of the center $x$, i.e., 
\[n\sum_{i=1}^{\left\lceil \frac{\sigma(2)}{8} \right\rceil}(2ki)\geq kn \left\lceil \frac{\sigma(2)}{8} \right\rceil^2=\Omega\big(\sigma(2)n^2\big).\] 
Hence, the social cost of $G$ is $\Omega\big(\sigma(2)n^2\big)$. Finally, the diameter of $G$ is equal to $2\left\lceil \frac{\sigma(2)}{8} \right\rceil \geq \frac{\sigma(2)}{4}$. 

We now prove that $G$ is \PS. First of all, no agent wants to delete an edge of $G$ as, by doing so, the agent would save $\sigma(2)$ by the removal of the edge, but the cost of another edge bought by the same agent would increase from $\sigma(2)$ to $\sigma(n)$. Since $n \geq 5$, we have that $\sigma(n) > 2\sigma(2)$, by Proposition~\ref{prop:basic_property_sigma}.

Now we show that no two agents $u$ and $v$ want to add the edge $uv$. Let $i$ and $j$ be the distances from $u$ and $v$ to the center $x$ of the spider, respectively. W.l.o.g., we assume that $j \geq i$. We divide the proof into two cases according to whether $u$ and $v$ are in the same leg of the spider or not.

If $u$ and $v$ are in the same leg of the spider,\footnote{We assume that this is the case if $i=0$.} then, by adding $uv$ the edge cost of $u$ would increase by $\frac{1}{2}\sigma(j-i)$, while the distance cost of $u$ would decrease by at most 
\begin{align*}
& 2\left(\left\lceil\frac{\sigma(2)}{8}\right\rceil-j\right)(j-i-1)+2\sum_{\ell=1}^{\left\lfloor\frac{j-i-1}{2}\right\rfloor}\ell\\
& \leq 2\left(\frac{\sigma(2)}{8}+1-j\right)(j-i)+\frac{1}{4}(j-i)^2\\
& \leq \frac{\sigma(2)}{4}(j-i)\leq \frac{1}{2}\sigma(j-i),
\end{align*}
where the last inequality holds by Proposition~\ref{prop:basic_property_sigma}, while the last but one inequality holds because $j \geq 2$. Hence, $u$ would not agree to adding the edge $uv$.

If $u$ and $v$ are in different legs of the spider, then, by adding $uv$ the edge cost of $u$ would increase by $\frac{1}{2}\sigma(i+j)$, while the distance cost of $u$ would decrease by at most
\[
\mu:=2\left(\left\lceil\frac{\sigma(2)}{8}\right\rceil-j\right)(j+i-1)+2\sum_{\ell=1}^{\left\lfloor\frac{j+i-1}{2}\right\rfloor}\ell.
\]
When $i,j=1$, the value $\mu$ is upper bounded by $\frac{\sigma(2)}{4}$ and therefore it is not convenient for $u$ to add the edge $uv$. When $j \geq 2$, i.e., $i+j\geq 3$, the value $\mu$ is upper bounded by ,
\begin{align*}
\mu & \leq 2\left(\frac{\sigma(2)}{8}+1-j\right)(j+i-1)+\frac{1}{2}(j+i-1)^2\\
		& \leq \frac{\sigma(2)}{4}(j+i) \leq \frac{1}{2}\sigma(i+j),
\end{align*}
where the last inequality holds by Proposition~\ref{prop:basic_property_sigma}. Hence, $u$ would not agree on adding  $uv$.
The claim follows.
\end{proof}

\noindent The \PS networks of Lemma~\ref{lemma:pairwise_stable_network_large_diameter}, depicted in Figure~\ref{fig:Lower_bound_PoA}, asymptotically reach the upper bound for the diameter of \PS networks.
Moreover, they allow us to prove asymptotically matching lower bounds to the PoA for the class of 2-edge-connected networks.
\begin{theorem}\label{thm:LB_PoA}
The PoA of SNCG is in $\Omega\left(\frac{\sigma(n)}{n\max\{\sigma(2),n\}}\right)$. For the class of 2-edge-connected networks the PoA is in $\Omega\big(\min\{\sigma(2),n\}\big)$.
\end{theorem}
\begin{proof}
First we prove the lower bound for the class of 2-edge-connected networks.
By Lemma~\ref{lemma:pairwise_stable_network_large_diameter}, there is a \PS network with $n=\Omega(\sigma(2))$ nodes and social cost in $\Omega\big(\sigma(2)n^2\big)$. Now, if we assume that $\sigma(2) >2$, then, by Theorem~\ref{thm:fan_and_clique_are_OPT}, we have that $\mathbf{F}_n$ is a social optimum. The social cost of $\mathbf{F}_n$ is at most 
\[2n(n-1)+(\sigma(2)-2)\frac{3}{2}(n-1)=O\big(n^2+\sigma(2)n\big).\]
Therefore, the PoA is in $\Omega\big(\min\{\sigma(2),n\}\big)$.

Concerning the lower bound for the general case, consider the modified fan graph $\mathbf{F}_n'$ from Theorem~\ref{thm:fan_and_clique_are_PS} for even $n$. This network is \PS for $\sigma(2) \geq 2$ and has a social cost of $O\big(\sigma(2)n+\sigma(n)\big)$. Thus, the PoA is in $\Omega\left(\frac{\sigma(n)}{n\max\{\sigma(2),n\}}\right)$.
\end{proof}

\noindent We conclude this section by showing bounds to the PoS.

\begin{theorem}\label{thm:PoS_PSN}
The PoS of the SNCG when $\sigma(2) \leq 2$ or $n$ is odd is 1. The PoS of the SNCG when $\sigma(2)>2$ and $n$ is even:
\begin{itemize}
\item[•] at most $\frac{11}{8}$ if $\sigma(3)\geq 6$ and $\sigma(2)\leq\frac{n}{2}-4$;
\item[•] at most $\frac{17}{12}$ if $\sigma(2)\geq \frac{2n}{3}$; 
\item[•] $\mathcal{O}\left(\frac{\sigma(n)}{n\max\big\{\sigma(2),n\big\}}\right)$, otherwise.
\end{itemize} 
\end{theorem}
\begin{proof}
For the cases in which $\sigma(2)\leq 2$ or $n$ is odd we have that the PoS is 1, since, from Theorem~\ref{thm:fan_and_clique_are_OPT} and Theorem~\ref{thm:fan_and_clique_are_PS}, there always exists a \PS network which is also a social optimum. 

It remains to prove the theorem statement for the case where $n$ is even and $\sigma(2)> 2$. We observe that the fan graph $\mathbf{F}_n$ is not \PS in this case as the node of degree 3 can remove the edge towards a node of degree 2. 
An upper bound for the PoA is delivered by the modified fan graph $\mathbf{F'}_n$ from the proof of Theorem~\ref{thm:fan_and_clique_are_PS}.
Then the PoS ratio is at most 
\begin{align*}
\frac{cost(\mathbf{F'}_n)}{cost(\mathbf{F}_n)} = \frac{(n-2)\frac{3}{2}\sigma(2)+\sigma(n) + 2n^2-5n+4}{(n-2)\frac{3}{2}\sigma(2)+2\sigma(2) + 2n^2-5n+2}
\in\mathcal{O}\left(\frac{\sigma(n)}{n\max\big\{\sigma(2),n\big\}}\right). 
\end{align*}
Unfortunately, the above upper bound is infinitely large for sufficiently large cost of a bridge $\sigma(n)$. 
However, under some additional assupmtions it is possible to show that the PoS is constant. 

We will show that for  $\sigma(3)\geq 6$ and $\sigma(2)\leq n/2$, there are \PS networks with no bridges, small diameter and containing 2-edges only.

\begin{figure}[h]
\centering
\includegraphics[width=12cm]{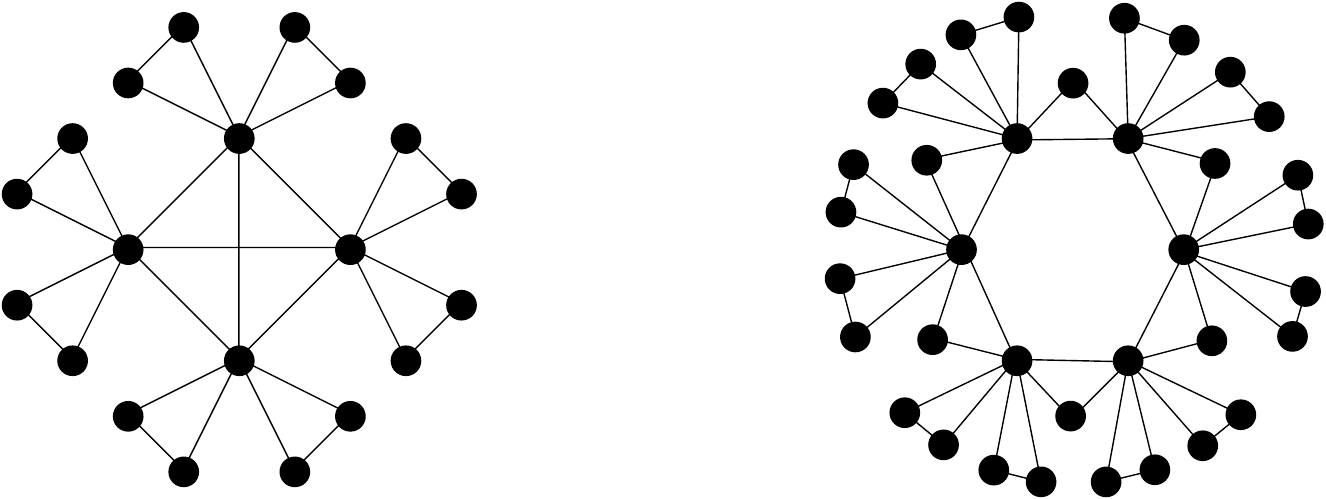}
\caption{Two instances of the \PS networks from the proof of Theorem~\ref{thm:PoS_PSN}. The left figure shows an example of \PS with 20 nodes for the case $\sigma(2)\leq \frac{n}{2}$ and $\sigma(3) \geq 6$. On the right side is an example with 36 nodes for the case $\sigma(2)\geq \frac{2n}{3}$.}
\label{fig:UB_PoS}
\end{figure} 

Consider a clique $\mathbf{K}_4$. Connect each node of the clique with $\frac{n-4}{8}$ edge-disjoint triangles (if $n\mod 8\neq 0$, we uniformly distribute triangles among nodes of the clique such that the number of incident triangles differs by at most 1 among any two nodes of the clique). 
See Figure~\ref{fig:UB_PoS}~(left) for an illustration of the construction.

We will now prove that this network is \PS. 
Consider an agent in the clique. 
She cannot deviate from her current strategy because deleting any of her incident edges from the clique increases her distance to at least $1+\frac{n-4}{4}-2$ nodes by 1 and decreases her edge cost by $\frac{\sigma(2)}{2}\leq \frac{n}{4}-2$. The addition of any edge to a node of a triangle would cost $\frac{\sigma(2)}{2}$ and would improve her distance by only 1. 
No agent located at a node of a triangle can delete any of her incident edges because it creates a high-cost bridge. 
The addition of an edge with a clique is also not profitable for the agent from the clique as we observed earlier. 
Also, the addition of an edge to a node from another triangle either costs $\sigma(2)/2$ and improves the distance cost by only 1, or it costs $\sigma(3)/2$ and improves the distance cost by 3. 
Hence, we conclude that the construction is \PS. 

The social cost of the construction asymptotically approaches the value $\frac{3}{2}\cdot(n-4)\sigma(2)+6\sigma(2) + \frac{11}{4}n^2$, while the social cost of a fan graph is $(n-2)\frac{3}{2}\sigma(2)+2\sigma(2) + 2n^2-5n+2$, i.e., the PoS goes to $\frac{11}{8}$ when $n$ goes to infinity.

Now we consider the case when $\sigma(2)> \frac{2n}{3}$. 
We construct the following \PS network $G$. 
Consider a cycle  of length 6 consisting of edge-disjoint triangles.
Each second node of the inner cycle is connected with $\frac{n-12}{12}$ edge-disjoint triangles (if $n\mod 12\neq 0$, we uniformly distribute the triangles among nodes of the cycle such that the number of triangles differs by at most 1). 
See Figure~\ref{fig:UB_PoS}~(right) for an illustration of the construction. 

We will show that $G$ is \PS for $\sigma(2)\geq 2n/3$. 
We start with the edge deletions. 
Clearly no agent can delete any edge from a triangle outside of the central 6-cycle because it creates a bridge of cost $\sigma(n)$. 
Also, an agent located in the central 6-cycle  cannot delete any of her edges since it either creates a bridge or increases her cost of another incident edge from $\sigma(2)$ to $\sigma(5)$. 
In the second case, the deletion of the edge decreases the agent's edge cost by $\sigma(2)$ and at the same time the edge cost is increased by $\sigma(5)/2 \geq \sigma(4)/2\geq \sigma(2)$ (by Proposition~\ref{prop:basic_property_sigma} and monotonicity of $\sigma$-function) since the edge cost of another edge is increased by this deletion. 

Now we will show that no pair of non-neighboring agents $u, v$ wants to add a new edge $uv$. 
Note, there are  two types of nodes in the central 6-cycle: nodes of degree 2 (outer cycle), and nodes of degree $4+\frac{n}{6}-2$ (inner cycle). 
Let $u$ be a node in the cycle with degree $\left(2+\frac{n}{6}\right)$. 
If $v$ is a similar node of the same degree, then one of two cases holds:
\begin{itemize}
\item[•] node $v$ is at distance 2 from $u$. Then $u$'s cost changes by $\frac{1}{2}\sigma(2) - 2\cdot\left(\frac{n}{6}-2\right) - 3$, since it improves distance by 1 to two nodes from the cycle and to two high-degree nodes. This changes agent $u$'s cost by $\frac{1}{2}\sigma(2) - \frac{n}{3} +1 > 0$, since $\sigma(2)\geq 2n/3$. Thus, the edge addition is not profitable for agent $u$.
\item[•] node $v$ is at distance 3 from $u$. Then $u$'s cost changes by $\frac{1}{2}\sigma(3) - 2\cdot\left(\frac{n}{6}-2\right) - 4$, since $u$ decreases her distance to only three nodes in the cycle and only one high-degree node by 2. The change of $u$' cost is $\frac{1}{2}\sigma(3) - \frac{n}{3} \geq 0$, because $\sigma(3)\geq \sigma(2)\geq 2n/3$. Hence the move is not profitable for $u$.
\end{itemize}
If $v$ is any other node in the network, i.e., any 2-degree node, then creating the edge $uv$ would cost at least $\sigma(2)/2$ and decrease the distance to at most $\frac{n}{6}$ nodes. 
Hence it is enough to have $\sigma(2)\geq \frac{n}{3}$ to prevent $u$ from creating any edge with a 2-degree node. 
Therefore, there is no improving edge addition between a $\left(2+\frac{n}{6}\right)$-degree node and any other node in~$G$.

Now we consider 2-degree nodes. 
Let $u$ be a 2-degree node from the central cycle. 
As we know from the previous case, $v$ cannot be $\left(2+\frac{n}{6}\right)$-degree node. 
If $v$ is another 2-degree node in the central cycle, the addition of the edge $(u,v)$ significantly decreases the distance to other nodes only if $v$ is at distance 3 or 4 from $u$. 
In the first case, $u$ pays $\frac{1}{2}\sigma(3)$ and improves her distances by at most $4+\frac{n}{6}-2$. 
Since $\frac{n}{6}+2<\frac{n}{3}<\frac{1}{2}\sigma(2)\leq \frac{1}{2}\sigma(3)$, this move is not profitable for $u$. 
In the second case, edge $(u,v)$ costs $\frac{1}{2}\sigma(4)$ and improves her distance cost by at most $2\cdot\left(\frac{n}{6}-2\right)+5$. Since   $\frac{1}{2}\sigma(4)\geq \sigma(2)\geq \frac{2n}{3} > \frac{n}{3}+1$, this move is not profitable.%1 for $u$. 

Analogously, if $v$ is a 2-degree node outside of the central cycle, then the maximum profit $u$ achieves if $v$ is at distance 4 from $u$. 
This move costs $\frac{1}{2}\sigma(4)$ for agent $u$ and decreases her distance to $\frac{n}{6}-3$ nodes by 1 and to three other nodes by 5 in total. 
Since  $\frac{1}{2}\sigma(4)\geq \sigma(2) > \frac{n}{6} - 2$, the edge will not be added. 

Finally, let $u$ be a 2-degree node that is not in the central cycle, i.e., is a 2-degree node in a vane.  
Clearly, creating an edge with any other 2-degree node connected to the same center is not profitable for $u$. 
Consider a 2-degree node $v$ in a vane at distance at least 3 from $u$. 
If the distance between $u$ and $v$ is 3, the addition of $uv$ improves the distance to at most two nodes but costs $\frac{1}{2}\sigma(3)$, and therefore is not profitable. 
If $v$ is at distance 4 from $u$, the addition of the edge $uv$ costs $\frac{1}{2}\sigma(4)\geq \sigma(2)\geq \frac{2n}{3}$ and improves the distance cost by at most $\frac{n}{3}+2 < \frac{2n}{3}$. 
If $v$ is at distance 5 from $u$, the addition of $uv$ costs $\frac{1}{2}\sigma(5) >  \sigma(2)\geq \frac{2n}{3}$ and improves the distance cost by at most $2\cdot\left(\frac{n}{6}-4\right) + 11 < \frac{2n}{3}$. 
Since we checked all possible agents' improving moves, $G$ is \PS. 

Network $G$ provides the following upper bound to the PoS
\[
\frac{cost(G)}{cost(\mathbf{F}_n)} = \frac{\frac{3\cdot 6}{2}\left(\frac{n}{6}-2\right)\sigma(2)+6\sigma(2)+\frac{17}{6}n^2-13n-258}{(n-2)\frac{3}{2}\sigma(2)+2\sigma(2) + 2n^2-5n+2}\leq \frac{17}{12}.
\] 
\end{proof}

\section{Dynamics of the SNCG}\label{sec:dynamic}
So far we have considered the SNCG as a one-shot-game, i.e., we only have specified the strategy space of the agents and then focused on analyzing the equilibria of the game. In this section we focus on a more constructive sequential view of the game. As our goal is to mimick real-world social networks, we want to study the process of how such networks evolve over time. For this, we consider some initial network and then we activate the agents sequentially. An active agent will try to decrease her current cost by adding (jointly with another agent) or deleting an edge in the current network. If this process converges to a state where no agent wants to add or delete edges, then a \PS network is found. Hence, such so-called improving move dynamics are a way for actually finding equilibrium states of a game. Such dynamics are guaranteed to converge if and only if the strategic game has the \emph{finite improvement property (FIP)}, i.e., if from any strategy vector any sequence of improving strategy changes must be finite. This is equivalent to the game being a potential game~\cite{MS96}. We start with the negative result that the convergence of improving move dynamics is not guaranteed for the SNCG.

\begin{theorem}
The SNCG does not have the FIP.
\end{theorem}
\begin{proof}
\begin{figure}[h]
\centering
\includegraphics[width=13cm]{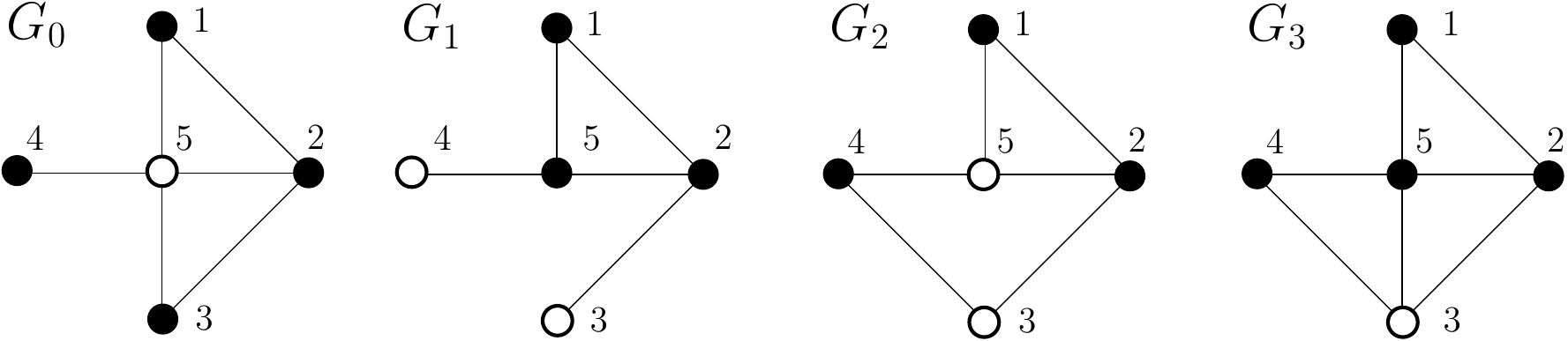}
\caption{A cyclic sequence of improving moves. The active player or pairs of players in each step are highlighted.}
\label{fig:BRC_PS}
\end{figure}
Let $\sigma$ be any function such that $2 < \sigma(2) < \frac{1}{2}\sigma(3)+1$ and $\sigma(3)< \frac{1}{2}\sigma(5)+2$.\footnote{Observe that there are several functions $\sigma$ that satisfy these additional constraints (for example $\sigma(x)=\frac{3}{2}x$).}
Consider the cyclic sequence of improving moves depicted on Figure~\ref{fig:BRC_PS} in which $G_{i}$ and $G_{(i+1)\mod 4}$ differ by exactly one edge. We will show that each step of the cycle is an improving response.\footnote{A careful reader will note that each improving move is also a best response.}
\begin{itemize}
\item[] $G_0\rightarrow G_1:$ agent $5$ deletes the edge towards agent $3$ because its edge cost decreases by $\frac{1}{2}\sigma(2)>1$ while its distance cost increases by 1. 
\item[] $G_1\rightarrow G_2:$ the edge between agents 3 and 4 is created. For each of the two agents 3 and 4, the distance cost decreases by 2 while the edge cost increases by $\sigma(3)-\frac{1}{2}\sigma(5)<2$. In fact, in $G_1$ both agents pay their share for a bridge of cost $\sigma(5)$ while in $G_2$ both agents pay their share for two edges of cost $\sigma(3)$ each. 
\item[] $G_2\rightarrow G_3:$ the edge between agents 3 and 5 is created. In fact, for agent 3 the edge cost varies by $\frac{3}{2}\sigma(2)-\sigma(3) \leq 0$ while the distance cost decreases by 1. For agent 5 the edge cost varies by $2\sigma(2)-\sigma(2)-\frac{1}{2}\sigma(3)=\sigma(2)-\frac{1}{2}\sigma(3)<1$ while the distance cost decreases by 1.
\item[] $G_3\rightarrow G_0:$ agent 3 deletes the edge towards agent 4 because its edge cost decreases by $\frac{1}{2}\sigma(2)>1$ while its distance cost increases by 1.\qedhere
\end{itemize}
\end{proof}

\noindent The above negative result for the sequential version of the SNCG should not be overrated. In fact, when simulating the sequential process it almost always converges to a \PS network. We will now discuss such simulations.
\subsection{Experimental Results}
We will illustrate that starting from a sparse initial network, the sequential version of the SNCG converges to a \PS network with real-world properties, like low diameter, high clustering and a power-law degree distribution. We will measure the clustering with the \emph{average local clustering coefficient (CC)}, that is a commonly used measure in Network Science~\cite{Bar16} \footnote{The clustering coefficient is the probability that two randomly chosen neighbors of a randomly chosen node in the network are neighbors themselves. More formally, let $\deg(v)$ denote the degree of $v$ in $G$ and let $\Delta(v)$ denote the number of triangles in $G$ that contain $v$ as a node. The \textit{local clustering coefficient} $CC(v)$ of node $v$ in $G$ is the probability that two randomly selected neighbors of $v$ are neighbors, i.e., $CC(v):=\frac{2\Delta(v)}{deg(v)(deg(v)-1)}$ if $deg(v) \geq 2$, and 0 otherwise. Clearly, $0\leq CC(v)\leq 1$. The CC of a network $G$ with $n$ nodes is the average of the local clustering coefficients over all nodes $v$, i.e., $CC(G)=\frac{1}{n}\sum_{v \in V}CC(v)$.}. Power-law degree distributions will be illustrated via log-log plots and a comparison with a perfect power-law distribution.

\begin{figure*}[t]
\centering
\begin{minipage}{0.22\textwidth}
\includegraphics[scale=0.16]{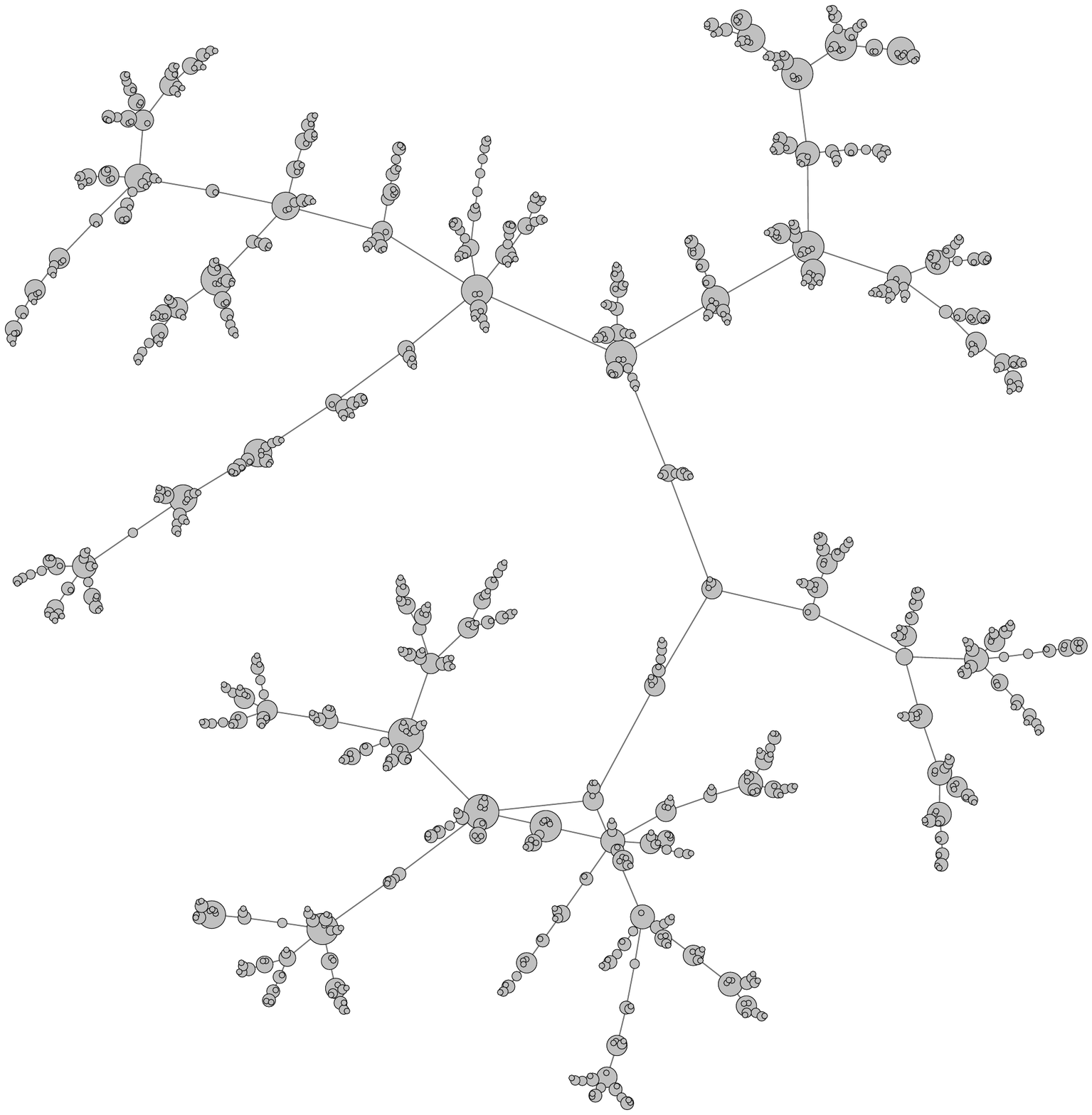}
\end{minipage}
\hfill
\begin{minipage}{0.25\textwidth}
\includegraphics[scale=0.19]{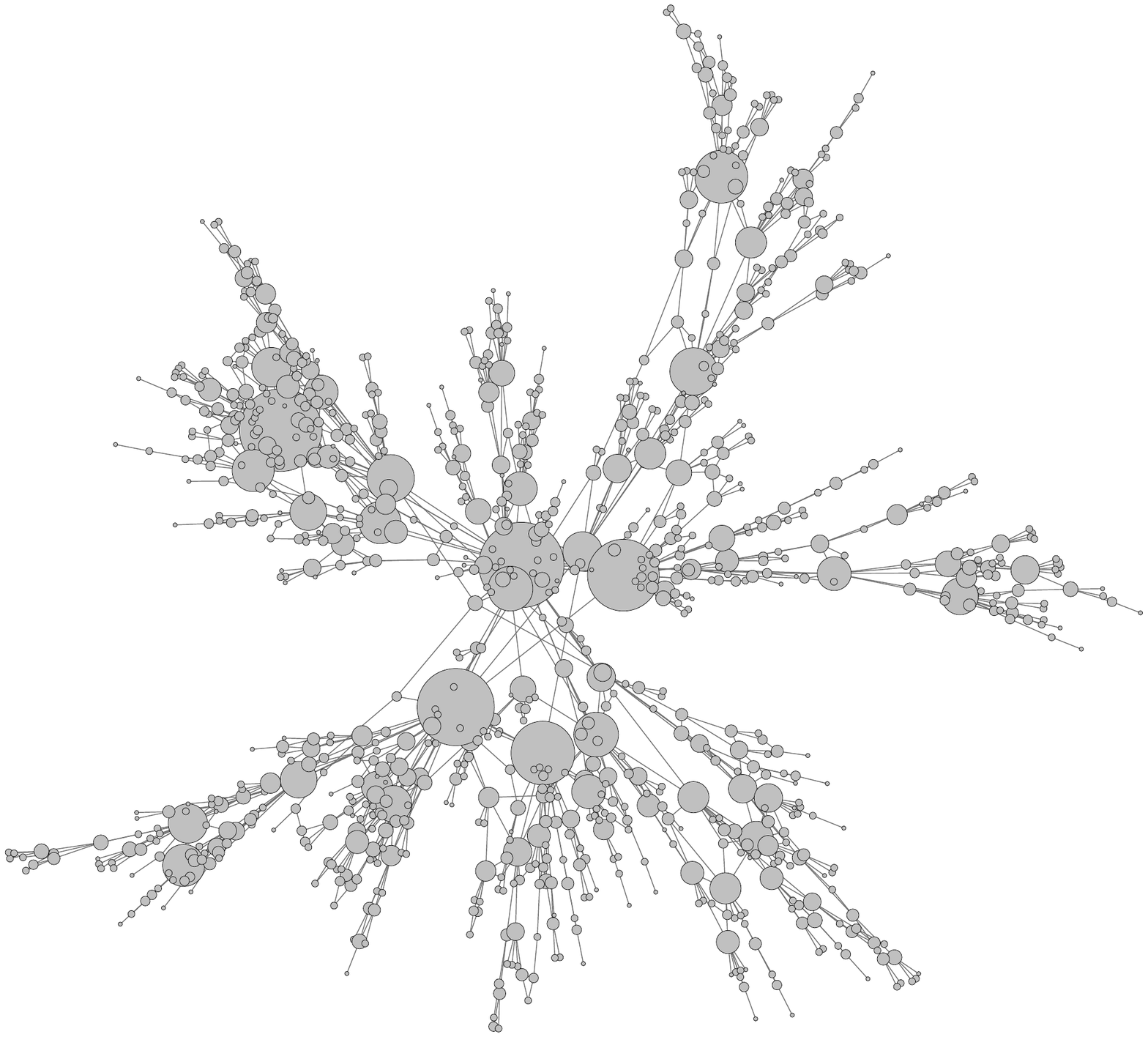}
\end{minipage}
\hfill
\begin{minipage}{0.25\textwidth}
\includegraphics[scale=0.19]{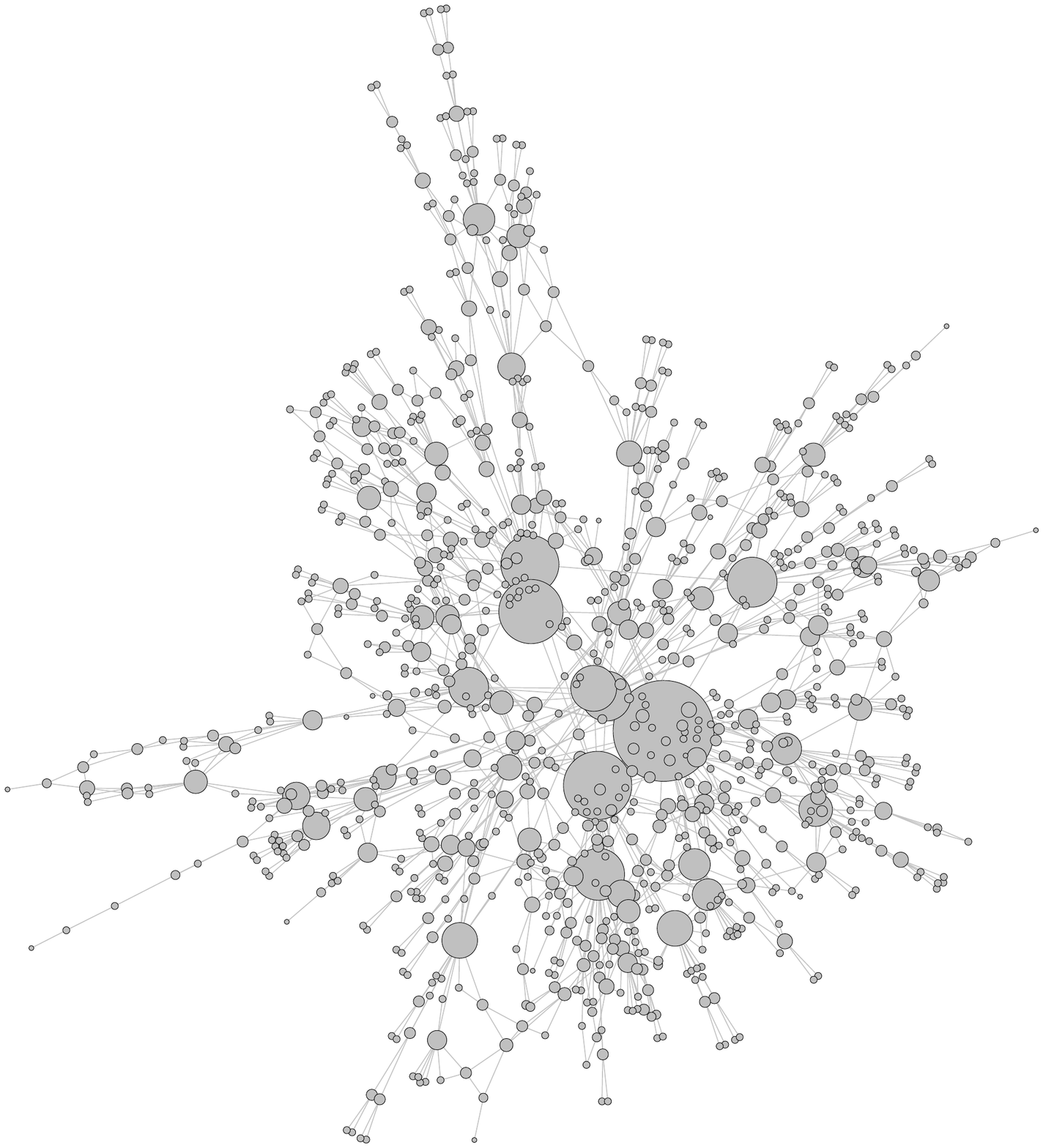}
\end{minipage}
\hfill
\begin{minipage}{0.25\textwidth}
\includegraphics[scale=0.19]{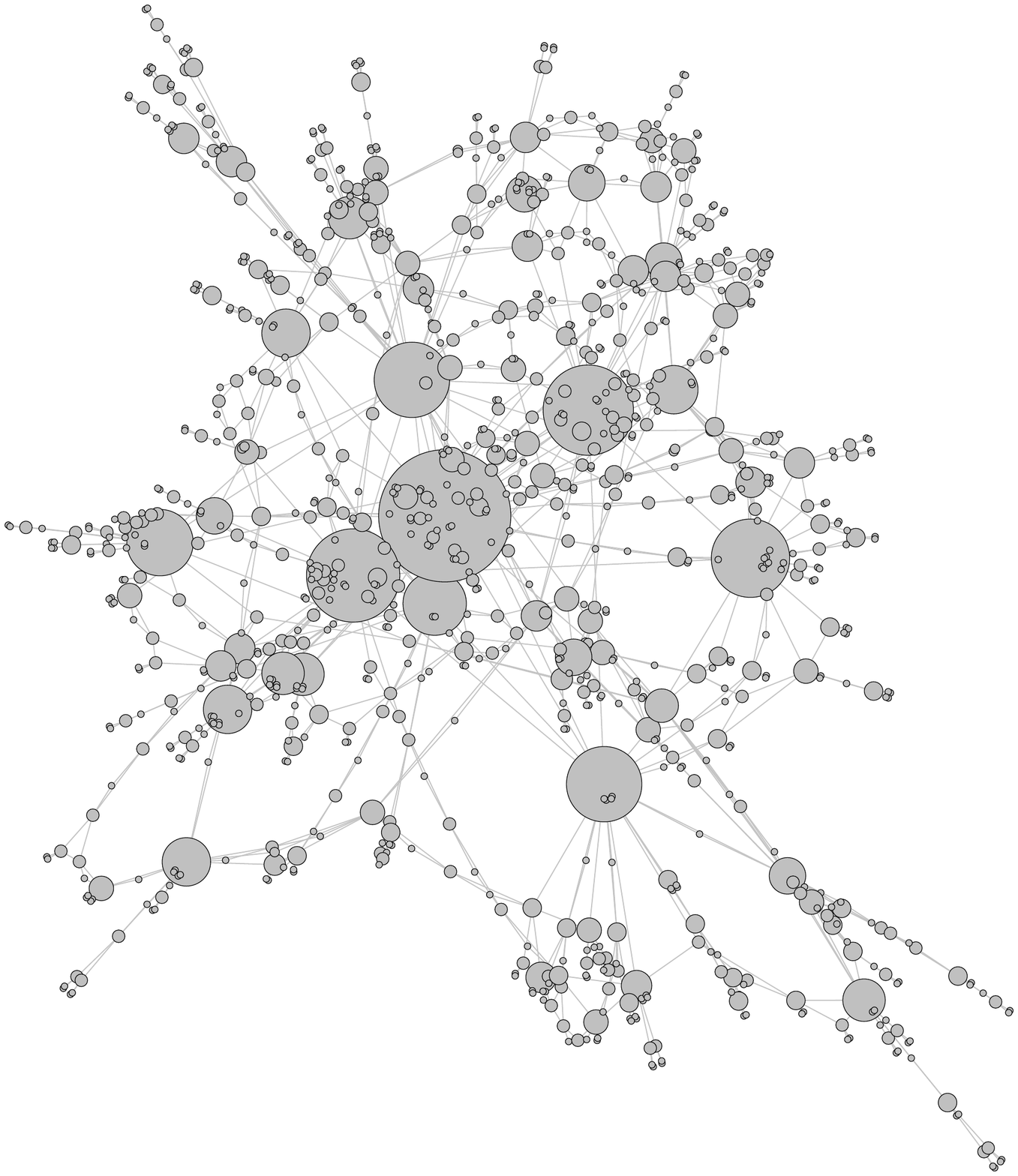}
\end{minipage}
\caption{Snapshots of networks obtained by the iterative best move dynamic starting from a random spanning tree with $n=1000$ and $\alpha=3$. Each plot from left to right shows the current network after 1000 steps each. The left plot shows the initial tree; the right plot shows is the final \PS network. The size of the nodes is proportional to the node degrees.}
\label{fig:graph_dynamic}
\end{figure*} 

For all experiments\footnote{The source code we used can be found at \url{https://github.com/melnan/distNCG.git}.} we choose $\sigma(x) = 2\log_2(n) \cdot x^\alpha$, where $n\in\mathbb{N}$ (the number of agents) and $\alpha\in \mathbb{R}$ (the exponent) are input parameters. Clearly, this function satisfies all constraints we have in the definition of the game, i.e., it is convex, monotone, and $\sigma(0)=0$. 

Note that by Proposition~\ref{prop:diam_PSN} the upper bound for the diameter of \PS networks is $\sigma(2)+2$ and thus we have to define $\sigma(2)$ to be growing with $n$ to avoid a constant diameter. Using $\sigma(x) = 2\log_2(n) \cdot x^\alpha$ as a proof-of-concept ensures a diameter upper bound of $\mathcal{O}(\log n)$ that is in line with the observed diameter bounds in many real-world networks~\cite{Bar16}. We emphasize that also other edge cost functions with similar properties yield similar results.

In each step of our simulations one agent is activated uniformly at random and this agent then performs the best possible edge addition (jointly with the other endpoint if the respective agent agrees) or edge deletion. If no such move exists then the agent is marked, otherwise the network is updated, and all marked agents become unmarked and we repeat. 
The process stops when all agents are marked. 
 
In our experiments, we always start from a sparse initial network, i.e., a cycle or a random spanning tree, to simulate an evolving social network, i.e., agents are initially connected with only very few other agents, and the number of new connections grows over time. See Figure~\ref{fig:graph_dynamic} for showcase snapshots from this process.  

 Additional experiments starting with sparse Erdös-Renyi random networks support our intuition that the network initialization does not matter as long as the networks are sparse and the average distances are large, i.e., the resulting stable structures have the same structural properties as starting from random trees or cycles. However, for example, starting from a star network yields drastically different results. Moreover, if the initial structure is a fan graph, the algorithm stops immediately since a fan is a stable network as stated in Theorem~\ref{thm:fan_and_clique_are_PS}. This shows that for the initial networks both sparseness and large average distances are crucial.

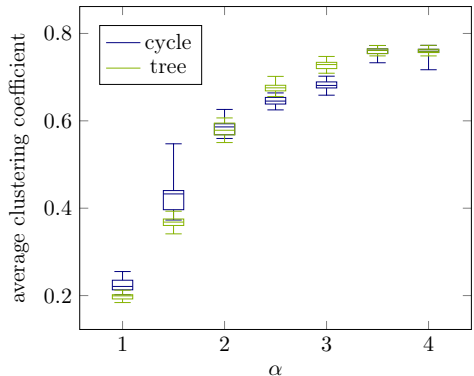
\begin{figure}[!ht]
\centering
\begin{minipage}{0.40\textwidth}
\resizebox {\textwidth} {!} {
\begin{tikzpicture}
\begin{axis}[
		legend style={
		at={(0.05,0.85)},
		anchor=west},
		xlabel= {$\alpha$},
		ylabel= {average clustering coefficient},
		boxplot/draw direction=y,
		baseline
	]
\addplot[blue!50!black, domain=1.1:1.2]{0.8};
\addlegendentry{cycle};
\addplot[lime!70!black, domain=1.1:1.2]{0.8};
\addlegendentry{tree};
	
\addplot+ [color = blue!50!black,solid,boxplot prepared = {box extend=0.2, draw position= 1, lower whisker = 0.201895, lower quartile = 0.213186, median = 0.220929, upper quartile = 0.235154, upper whisker = 0.255019},]coordinates{}; 

\addplot+ [color = blue!50!black,solid,boxplot prepared = {box extend=0.2, draw position= 1.5, lower whisker =0.373011, lower quartile = 0.396586, median = 0.432716, upper quartile = 0.440368, upper whisker = 0.547433},]coordinates{}; 

\addplot+ [color = blue!50!black,solid,boxplot prepared = {box extend=0.2, draw position= 2, lower whisker =0.559647, lower quartile = 0.568128, median = 0.586087, upper quartile = 0.593227, upper whisker = 0.626096},]coordinates{};

\addplot+ [color = blue!50!black,solid,boxplot prepared = {box extend=0.2, draw position= 2.5, lower whisker =0.625038, lower quartile = 0.638440, median = 0.645200, upper quartile = 0.653596, upper whisker = 0.663893},]coordinates{};

\addplot+ [color = blue!50!black,solid,boxplot prepared = {box extend=0.2, draw position= 3, lower whisker =0.658702, lower quartile = 0.675279, median = 0.681046, upper quartile = 0.689108, upper whisker = 0.702182},]coordinates{};

\addplot+ [color = blue!50!black,solid,boxplot prepared = {box extend=0.2, draw position= 3.5, lower whisker =0.732935, lower quartile = 0.753958, median = 0.761218, upper quartile = 0.764726, upper whisker = 0.772473},]coordinates{};

\addplot+ [color = blue!50!black,solid,boxplot prepared = {box extend=0.2, draw position= 4, lower whisker =0.716947, lower quartile = 0.756336, median = 0.759216, upper quartile = 0.763776, upper whisker = 0.772799},]coordinates{};

%tree
\addplot+ [color = lime!70!black,solid,boxplot prepared = {box extend=0.2, draw position= 1, lower whisker = 0.184032, lower quartile = 0.192299, median = 0.199339, upper quartile = 0.201977, upper whisker = 0.212370},]coordinates{}; 

\addplot+ [color = lime!70!black,solid,boxplot prepared = {box extend=0.2, draw position= 1.5, lower whisker =0.341295, lower quartile = 0.360498, median = 0.368204, upper quartile = 0.375317, upper whisker = 0.393142},]coordinates{}; 

\addplot+ [color = lime!70!black,solid,boxplot prepared = {box extend=0.2, draw position= 2, lower whisker =0.550264, lower quartile = 0.567989, median = 0.578451, upper quartile = 0.594834, upper whisker = 0.606655},]coordinates{};

\addplot+ [color = lime!70!black,solid,boxplot prepared = {box extend=0.2, draw position= 2.5, lower whisker =0.655603, lower quartile = 0.668577, median = 0.675514, upper quartile = 0.681368, upper whisker = 0.701735},]coordinates{};

\addplot+ [color = lime!70!black,solid,boxplot prepared = {box extend=0.2, draw position= 3, lower whisker =0.708975, lower quartile = 0.719598, median = 0.728825, upper quartile = 0.733889, upper whisker = 0.747373},]coordinates{};

\addplot+ [color = lime!70!black,solid,boxplot prepared = {box extend=0.2, draw position= 3.5, lower whisker =0.748643, lower quartile = 0.753958, median = 0.761218, upper quartile = 0.764726, upper whisker = 0.772473},]coordinates{};

\addplot+ [color = lime!70!black,solid,boxplot prepared = {box extend=0.2, draw position= 4, lower whisker =0.748614, lower quartile = 0.756336, median = 0.759216, upper quartile = 0.763776, upper whisker = 0.772799},]coordinates{};

\end{axis}
\end{tikzpicture}
}
\end{minipage}
\caption{Average clustering coefficient of \PS networks obtained by the best move dynamic for $n=1000$ over 20 runs with $\sigma(x) = 18 x^\alpha$. Blue: results of the process starting from a cycle; green: starting from a random tree.}
\label{plot:CC}
\end{figure}

Figure~\ref{plot:CC} shows the box-and-whiskers plot for the average clustering coefficient of the pairwise stable networks obtained by the algorithm for $n=1000$ with respect to the value of the power coefficient $\alpha$. 
The upper and lower whiskers show the maximal and the minimal average clustering coefficient over 20 runs. 
The bottom and top of the boxes are the first and the third quartiles; the middle lines are the median values. 
The plot explicitly shows that pairwise networks generated by the best move dynamic for a polynomial edge-cost function have a high clustering coefficient. 
The results indicate that the clustering coefficient correlates with the power coefficient $\alpha$.

\begin{figure}[!ht]
\begin{minipage}{0.24\textwidth}
\resizebox {\textwidth} {!} {
\begin{tikzpicture}
\begin{axis}[
xmode=log,
ymode=log,
		legend style={
		at={(0.55,0.75)},
		anchor=west},
		xlabel= {degree},
		ylabel= {local clustering coefficient},
		baseline
	]

-\addplot+ [color = blue!50!black,only marks, mark=*] coordinates{  (3,9) (4,19) (5,16) (6,26) (7,37) (8,48) (9,80) (10,85) (11,111) (12,155) (13,181) (14,159) (15,183) (16,161) (17,180) (18,152) (19,121) (20,114) (21,86) (22,65) (23,70) (24,58) (25,37) (26,30) (27,38) (28,31) (29,25) (30,26) (31,20) (32,23) (33,11) (34,16) (35,12) (36,18) (37,15) (38,15) (39,19) (40,11) (41,13) (42,14) (43,17) (44,16) (45,17) (46,15) (47,12) (48,14) (49,7) (50,10) (51,10) (52,6) (53,10) (54,11) (55,12) (56,10) (57,7) (58,15) (59,12) (60,8) (61,14) (62,19) (63,12) (64,5) (65,10) (66,9) (67,7) (68,6) (69,5) (70,11) (71,8) (72,4) (73,5) (74,10) (75,9) (76,5) (77,7) (78,3) (79,8) (80,8) (81,10) (82,5) (83,5) (84,2) (85,11) (86,5) (87,2) (88,6) (89,7) (90,5) (91,3) (92,8) (93,3) (94,4) (95,6) (96,4) (97,2) (98,3) (99,4) (100,1) (101,1) (102,3) (103,3) (104,3) (107,3) (108,1) (109,3) (110,2) (111,3) (112,4) (113,3) (114,1) (115,2) (116,2) (117,1) (118,1) (121,1) (122,1) (127,1) (395,1)}; 
\addlegendentry{cycle, $\alpha = 1$};
\addplot+ [color = lime!70!black,only marks, mark=+] coordinates{  (3,9) (4,22) (5,16) (6,10) (7,20) (8,25) (9,38) (10,67) (11,117) (12,133) (13,143) (14,143) (15,175) (16,166) (17,126) (18,160) (19,133) (20,111) (21,99) (22,93) (23,72) (24,57) (25,80) (26,52) (27,47) (28,36) (29,24) (30,37) (31,22) (32,21) (33,19) (34,17) (35,17) (36,15) (37,20) (38,23) (39,11) (40,12) (41,13) (42,19) (43,14) (44,25) (45,13) (46,12) (47,12) (48,13) (49,13) (50,9) (51,18) (52,14) (53,17) (54,16) (55,11) (56,16) (57,11) (58,13) (59,8) (60,11) (61,10) (62,8) (63,7) (64,12) (65,3) (66,11) (67,15) (68,7) (69,10) (70,10) (71,10) (72,16) (73,6) (74,14) (75,11) (76,9) (77,7) (78,5) (79,7) (80,9) (81,12) (82,7) (83,5) (84,4) (85,1) (86,6) (87,9) (88,9) (89,5) (90,4) (91,3) (92,5) (93,6) (94,4) (95,5) (96,3) (97,2) (98,1) (99,3) (100,3) (101,4) (102,3) (103,2) (104,2) (105,1) (106,2) (107,3) (108,4) (109,5) (111,2) (112,1) (113,3) (114,4) (116,2) (117,2) (118,2) (119,1) (120,1) (121,1) (122,1) (124,1) (125,1) (128,1) (146,1)}; 
\addlegendentry{tree, $\alpha = 1$};
\addplot[black, domain = 4:100.5]{20000*x^(-2.1)};
\addlegendentry{ $y \sim k^{-2.1}$};

\end{axis}
\end{tikzpicture}
}
%\caption*{$\beta = 2$.}
\end{minipage}
\hfill
\begin{minipage}{0.24\textwidth}
\resizebox {\textwidth} {!} {
\begin{tikzpicture}
\begin{axis}[
xmode=log,
ymode=log,
		legend style={
		at={(0.55,0.75)},
		anchor=west},
		xlabel= {degree},
		ylabel= {number of nodes},
		baseline
	]

\addplot+ [color = blue!50!black,only marks, mark=*] coordinates{  (1,1) (2,738) (3,223) (5,429) (6,201) (7,176) (8,104) (9,90) (10,63) (11,47) (12,43) (13,38) (14,31) (15,37) (16,21) (17,24) (18,29) (19,25) (20,36) (21,28) (22,31) (23,23) (24,14) (25,26) (26,32) (27,29) (28,26) (29,28) (30,23) (31,30) (32,21) (33,23) (34,23) (35,26) (36,23) (37,17) (38,17) (39,20) (40,23) (41,18) (42,16) (43,13) (44,17) (45,18) (46,14) (47,14) (48,6) (49,7) (50,4) (51,6) (52,9) (53,5) (54,3) (55,2) (56,2) (59,3) (60,2) (73,1) (89,1)}; 
\addlegendentry{cycle, $\alpha = 2$};
\addplot+ [color = lime!70!black,only marks, mark=+] coordinates{  (1,1) (2,697) (4,242) (5,399) (6,189) (7,165) (8,85) (9,65) (10,74) (11,61) (12,43) (13,41) (14,37) (15,42) (16,36) (17,37) (18,26) (19,30) (20,30) (21,34) (22,20) (23,20) (24,39) (25,30) (26,35) (27,22) (28,32) (29,20) (30,28) (31,26) (32,30) (33,19) (34,17) (35,26) (36,28) (37,35) (38,24) (39,26) (40,21) (41,14) (42,18) (43,14) (44,17) (45,11) (46,12) (47,10) (48,15) (49,4) (50,15) (51,4) (52,4) (53,9) (54,2) (55,5) (56,2) (57,2) (58,1) (59,2) (60,1) (61,1) (62,1) (64,1) (65,2) (73,1)}; 
\addlegendentry{tree, $\alpha = 2$};

\addplot[black, domain = 3:50.5]{10000*x^(-2.3)};
\addlegendentry{ $y \sim k^{-2.3}$};

\end{axis}
\end{tikzpicture}
}
%\caption*{$\beta = 2$.}
\end{minipage}
\hfill
\begin{minipage}{0.24\textwidth}
\resizebox {\textwidth} {!} {
\begin{tikzpicture}
\begin{axis}[
xmode=log,
ymode=log,
%log ticks with fixed point,
		legend style={
		at={(0.55,0.75)},
		anchor=west},
		xlabel= {degree},
		ylabel= {number of nodes},
		baseline
	]

\addplot+ [color = blue!50!black,only marks, mark=*] coordinates{  (1,1) (2,1362) (3,220) (4,696) (5,111) (6,246) (8,60) (9,77) (10,29) (11,38) (12,22) (13,19) (14,8) (15,10) (16,12) (17,6) (18,4) (19,10) (20,5) (21,6) (22,6) (23,4) (24,4) (25,8) (26,2) (27,2) (28,5) (29,1) (30,4) (31,2) (32,2) (33,3) (34,4) (36,1) (37,1) (40,2) (42,1) (45,1) (52,1) (53,1) (56,1) (57,1) (60,1)}; 
\addlegendentry{cycle, $\alpha = 3$};
\addplot+ [color = lime!70!black,only marks, mark=+] coordinates{  (1,1) (2,1561) (3,183) (4,682) (6,70) (7,209) (8,27) (9,75) (10,20) (11,39) (12,9) (13,29) (14,6) (15,10) (16,9) (17,12) (18,6) (19,4) (20,1) (21,8) (22,1) (23,3) (24,4) (25,3) (26,1) (27,2) (28,3) (30,1) (32,1) (33,1) (34,2) (37,2) (38,1) (39,1) (40,1) (41,1) (43,1) (45,1) (47,1) (48,1) (49,1) (51,1) (52,3) (54,1) (55,1)}; 
\addlegendentry{tree, $\alpha = 3$};

\addplot[black, domain = 2:50.5]{10000*x^(-2.5)};
\addlegendentry{ $y \sim k^{-2.5}$};
\end{axis}
\end{tikzpicture}
}
%\caption*{$\beta = 3$.}
\end{minipage}
\hfill
\begin{minipage}{0.24\textwidth}
\resizebox {\textwidth} {!} {
\begin{tikzpicture}
\begin{axis}[
xmode=log,
ymode=log,
%log ticks with fixed point,
		legend style={
		at={(0.55,0.75)},
		anchor=west},
		xlabel= {degree},
		ylabel= {number of nodes},
		baseline
	]
\addplot+ [color = blue!50!black,only marks, mark=*] coordinates{  (1,1) (2,1792)  (4,723) (5,5) (6,269) (8,89) (10,44) (12,22) (13,2) (14,6) (16,6) (17,1) (18,4) (20,4) (22,1) (23,1) (24,2) (25,1) (28,2) (32,1) (35,1) (37,1) (39,1) (42,1)}; 
\addlegendentry{cycle, $\alpha = 4$};

\addplot+ [color = lime!70!black,only marks, mark=+] coordinates{  (2,1931) (4,704) (5,1) (6,207) (8,68) (10,29) (12,22) (14,11) (16,8) (17,1) (18,5) (20,3) (21,1) (22,1) (31,1) (35,1) (39,1) (44,1) (55,1)}; 
\addlegendentry{tree, $\alpha = 4$};

\addplot[black, domain = 2.1:40.5]{12000*x^(-2.7)};
\addlegendentry{ $y \sim k^{-2.5}$};

\end{axis}
\end{tikzpicture}
}
%\caption*{$\beta = 3$.}
\end{minipage}
\caption{Log-log plot of the degree distribution  of \PS networks obtained by the best move dynamic for $n=3000$ with $\sigma(x) = 18 x^\alpha$. Blue: results for the process starting from a cycle; green: starting from a random tree. Black: a fitted perfect power law distribution.}
\label{plot:deg_distr}
\end{figure}
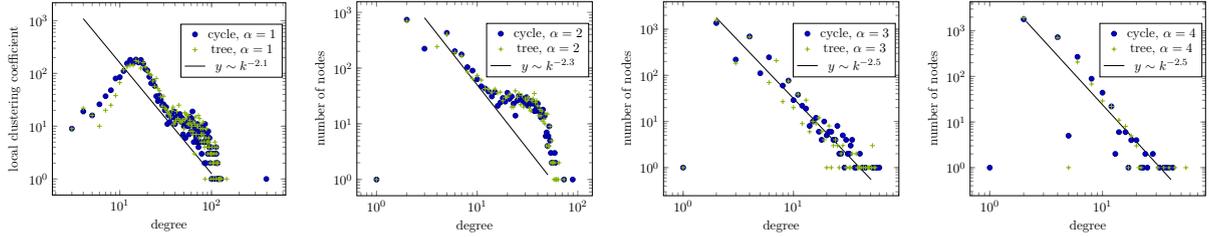

Figure~\ref{plot:deg_distr} shows a degree distribution for the resulting \PS networks for $n=3000$.  
We supplemented each plot with a plot of a perfect power-law distribution $P(k)\sim k^{-\gamma}$. 
All our experiments show that the power-law exponent $\gamma$ is between 2 and 3, which indicates that our generated \PS networks are indeed scale-free.

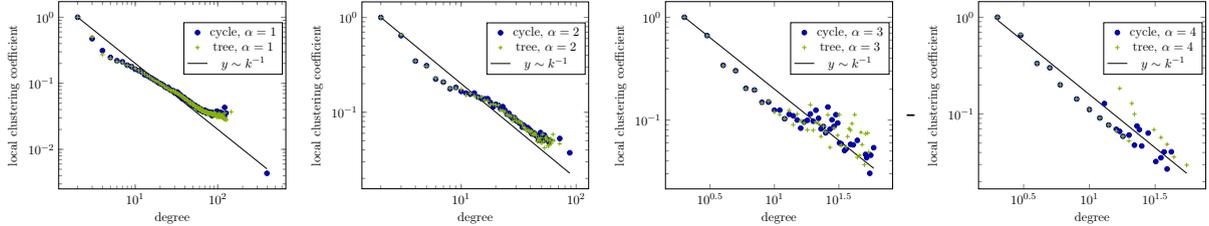
\begin{figure}[!ht]
\begin{minipage}{0.24\textwidth}
\resizebox {\textwidth} {!} {
\begin{tikzpicture}
\begin{axis}[
xmode=log,
ymode=log,
%    log ticks with fixed point,
		legend style={
		at={(0.55,0.75)},
		anchor=west},
		xlabel= {degree},
		ylabel= {local clustering coefficient},
		baseline
	]
	\addplot+ [color = blue!50!black,only marks, mark=*] coordinates{  (2,1.000000) (3,0.473684) (4,0.312500) (5,0.246154) (6,0.219820) (7,0.211309) (8,0.188393) (9,0.179085) (10,0.163363) (11,0.156364) (12,0.147916) (13,0.134736) (14,0.130367) (15,0.123396) (16,0.118241) (17,0.112326) (18,0.108086) (19,0.102390) (20,0.099327) (21,0.096850) (22,0.092703) (23,0.088047) (24,0.085586) (25,0.087111) (26,0.081296) (27,0.078669) (28,0.077989) (29,0.076923) (30,0.071609) (31,0.072464) (32,0.074047) (33,0.071970) (34,0.065508) (35,0.067040) (36,0.064656) (37,0.058759) (38,0.060717) (39,0.061219) (40,0.058383) (41,0.056185) (42,0.055817) (43,0.055302) (44,0.053414) (45,0.053737) (46,0.048390) (47,0.051011) (48,0.047112) (49,0.048639) (50,0.047837) (51,0.045752) (52,0.045701) (53,0.046312) (54,0.043967) (55,0.044444) (56,0.042857) (57,0.042356) (58,0.043305) (59,0.039962) (60,0.039508) (61,0.041012) (62,0.040323) (63,0.038812) (64,0.039683) (65,0.039263) (66,0.037363) (67,0.037389) (68,0.037665) (69,0.037976) (70,0.039286) (71,0.037726) (72,0.036228) (73,0.035388) (74,0.037476) (75,0.036685) (76,0.035739) (77,0.036341) (78,0.036131) (79,0.036108) (80,0.036709) (81,0.036543) (82,0.034267) (83,0.035851) (84,0.034945) (85,0.032493) (86,0.034063) (87,0.035730) (88,0.035602) (89,0.034270) (90,0.034457) (91,0.033364) (92,0.035276) (93,0.034712) (94,0.035194) (95,0.034043) (96,0.034539) (97,0.037872) (98,0.034662) (99,0.036281) (100,0.035556) (101,0.034125) (102,0.036821) (103,0.034203) (106,0.035340) (107,0.032446) (108,0.036402) (109,0.034064) (110,0.033917) (111,0.034070) (112,0.035071) (113,0.034766) (114,0.032060) (115,0.032265) (116,0.033133) (117,0.031536) (120,0.031232) (121,0.043113) (126,0.035429) (394,0.004340)}; 
\addlegendentry{cycle, $\alpha = 1$};

\addplot+ [color = lime!70!black,only marks, mark=+] coordinates{  (2,1.000000) (3,0.500000) (4,0.270833) (5,0.240000) (6,0.213333) (7,0.205714) (8,0.187030) (9,0.180763) (10,0.167522) (11,0.156528) (12,0.146006) (13,0.138784) (14,0.129796) (15,0.124900) (16,0.117328) (17,0.111075) (18,0.106737) (19,0.105263) (20,0.098724) (21,0.095033) (22,0.092893) (23,0.086124) (24,0.087636) (25,0.081923) (26,0.080982) (27,0.077161) (28,0.077601) (29,0.077619) (30,0.072309) (31,0.073835) (32,0.070352) (33,0.069742) (34,0.067317) (35,0.066779) (36,0.064127) (37,0.063324) (38,0.062330) (39,0.058704) (40,0.057791) (41,0.057510) (42,0.055915) (43,0.055149) (44,0.053098) (45,0.052525) (46,0.049275) (47,0.051946) (48,0.048554) (49,0.047902) (50,0.048163) (51,0.046667) (52,0.047201) (53,0.043995) (54,0.044470) (55,0.043056) (56,0.042857) (57,0.043811) (58,0.041591) (59,0.040115) (60,0.042373) (61,0.040505) (62,0.041928) (63,0.038701) (64,0.037698) (65,0.038374) (66,0.039005) (67,0.037023) (68,0.037401) (69,0.037383) (70,0.036439) (71,0.036997) (72,0.036385) (73,0.037671) (74,0.035475) (75,0.035916) (76,0.034787) (77,0.033971) (78,0.035583) (79,0.035017) (80,0.035311) (81,0.034524) (82,0.034568) (83,0.035777) (84,0.034997) (85,0.033613) (86,0.033926) (87,0.033592) (88,0.034378) (89,0.032623) (90,0.035289) (91,0.035897) (92,0.034281) (93,0.033602) (94,0.033585) (95,0.033819) (96,0.034210) (97,0.030713) (98,0.031980) (99,0.032983) (100,0.032323) (101,0.033465) (102,0.030868) (103,0.034457) (104,0.030807) (105,0.031319) (106,0.032165) (107,0.034033) (108,0.031533) (110,0.031359) (111,0.033251) (112,0.030459) (113,0.032356) (115,0.032265) (116,0.029610) (117,0.029841) (118,0.030711) (119,0.029768) (120,0.028571) (121,0.035399) (123,0.028655) (124,0.028849) (127,0.027747) (145,0.036973)}; 
\addlegendentry{tree, $\alpha = 1$};

\addplot[black, domain = 2:394]{2*x^(-1)};
\addlegendentry{ $y \sim k^{-1}$};
\end{axis}
\end{tikzpicture}
}
%\caption{$\beta = 1$.}
\end{minipage}
\hfill
\begin{minipage}{0.24\textwidth}
\resizebox {\textwidth} {!} {
\begin{tikzpicture}
\begin{axis}[
xmode=log,
ymode=log,
		legend style={
		at={(0.55,0.75)},
		anchor=west},
		xlabel= {degree},
		ylabel= {local clustering coefficient},
		baseline
	]

\addplot+ [color = blue!50!black,only marks, mark=*] coordinates{  (1,0.000000) (2,1.000000) (3,0.642751) (4,0.346542) (5,0.309950) (6,0.223864) (7,0.207875) (8,0.176984) (9,0.182540) (10,0.165957) (11,0.158985) (12,0.155104) (13,0.159223) (14,0.144639) (15,0.139229) (16,0.138542) (17,0.138945) (18,0.126536) (19,0.113873) (20,0.121428) (21,0.118433) (22,0.108601) (23,0.114625) (24,0.104933) (25,0.095313) (26,0.097400) (27,0.090182) (28,0.086451) (29,0.087599) (30,0.081839) (31,0.080082) (32,0.076964) (33,0.076169) (34,0.069519) (35,0.072853) (36,0.069935) (37,0.067479) (38,0.068777) (39,0.064308) (40,0.064672) (41,0.059451) (42,0.058787) (43,0.060713) (44,0.058786) (45,0.056205) (46,0.056660) (47,0.050725) (48,0.053571) (49,0.061012) (50,0.049660) (51,0.053421) (52,0.050226) (53,0.057329) (54,0.052061) (55,0.052862) (58,0.048195) (59,0.050555) (72,0.052817) (88,0.037356)}; 
\addlegendentry{cycle, $\alpha = 2$};
\addplot+ [color = lime!70!black,only marks, mark=+] coordinates{  (1,0.000000) (2,1.000000) (3,0.655648) (4,0.346700) (5,0.305820) (6,0.227879) (7,0.208963) (8,0.179670) (9,0.178679) (10,0.174863) (11,0.151374) (12,0.158537) (13,0.148995) (14,0.142857) (15,0.145238) (16,0.134009) (17,0.131505) (18,0.120697) (19,0.122222) (20,0.118576) (21,0.119048) (22,0.104762) (23,0.103476) (24,0.101691) (25,0.094381) (26,0.092587) (27,0.090812) (28,0.080820) (29,0.088054) (30,0.082228) (31,0.078853) (32,0.079266) (33,0.078097) (34,0.074524) (35,0.069148) (36,0.067982) (37,0.064064) (38,0.064449) (39,0.063685) (40,0.061905) (41,0.062669) (42,0.057657) (43,0.058693) (44,0.060061) (45,0.056734) (46,0.049372) (47,0.053161) (48,0.050310) (49,0.053458) (50,0.047347) (51,0.058431) (52,0.049187) (53,0.042090) (54,0.053249) (55,0.054209) (56,0.050974) (57,0.043860) (58,0.045070) (59,0.049679) (60,0.052542) (61,0.046995) (63,0.058372) (64,0.050347) (72,0.046557)}; 
\addlegendentry{tree, $\alpha = 2$};

\addplot[black, domain = 2:88]{2*x^(-1)};
\addlegendentry{ $y \sim k^{-1}$};

\end{axis}
\end{tikzpicture}
}
%\caption*{$\beta = 2$.}
\end{minipage}
\hfill
\begin{minipage}{0.24\textwidth}
\resizebox {\textwidth} {!} {
\begin{tikzpicture}
\begin{axis}[
xmode=log,
ymode=log,
%log ticks with fixed point,
		legend style={
		at={(0.55,0.75)},
		anchor=west},
		xlabel= {degree},
		ylabel= {local clustering coefficient},
		baseline
	]

\addplot+ [color = blue!50!black,only marks, mark=*] coordinates{  (1,0.000000) (2,1.000000) (3,0.663637) (4,0.340517) (5,0.301802) (6,0.203523) (7,0.196032) (8,0.148887) (9,0.149425) (10,0.124561) (11,0.124794) (12,0.102871) (13,0.113782) (14,0.109890) (15,0.100000) (16,0.083333) (17,0.095588) (18,0.099346) (19,0.114620) (20,0.097368) (21,0.124603) (22,0.103896) (23,0.082016) (24,0.086504) (25,0.075000) (26,0.133846) (27,0.081482) (28,0.100529) (29,0.086823) (30,0.112644) (31,0.093548) (32,0.059812) (33,0.058712) (35,0.050420) (36,0.052381) (39,0.058030) (41,0.057317) (44,0.063425) (51,0.046275) (52,0.042986) (55,0.030303) (56,0.045455) (59,0.053770)}; 
\addlegendentry{cycle, $\alpha = 3$};
\addplot+ [color = lime!70!black,only marks, mark=+] coordinates{  (1,0.000000) (2,1.000000) (3,0.661203) (4,0.337976) (5,0.300000) (6,0.201595) (7,0.195767) (8,0.146190) (9,0.152778) (10,0.119658) (11,0.113131) (12,0.102926) (13,0.123932) (14,0.090110) (15,0.114286) (16,0.098611) (17,0.093137) (18,0.080065) (19,0.140351) (20,0.079605) (21,0.095238) (22,0.092352) (23,0.069170) (24,0.086956) (25,0.076667) (26,0.112308) (27,0.054131) (29,0.083744) (31,0.070968) (32,0.112903) (33,0.056818) (36,0.054762) (37,0.069069) (38,0.082504) (39,0.079622) (40,0.092308) (42,0.112660) (44,0.043340) (46,0.045411) (47,0.075856) (48,0.057624) (50,0.036735) (51,0.073464) (53,0.074746) (54,0.048917)}; 
\addlegendentry{tree, $\alpha = 3$};

\addplot[black, domain = 2:59]{2*x^(-1)};
\addlegendentry{ $y \sim k^{-1}$};
\end{axis}
\end{tikzpicture}
}
%\caption*{$\beta = 3$.}
\end{minipage}
-\hfill
\begin{minipage}{0.24\textwidth}
\resizebox {\textwidth} {!} {
\begin{tikzpicture}
\begin{axis}[
xmode=log,
ymode=log,
%log ticks with fixed point,
		legend style={
		at={(0.55,0.75)},
		anchor=west},
		xlabel= {degree},
		ylabel= {local clustering coefficient},
		baseline
	]
\addplot+ [color = blue!50!black,only marks, mark=*] coordinates{  (1,0.000000) (2,1.000000) (3,0.650000) (4,0.333333) (5,0.300000) (6,0.200000) (8,0.142857) (10,0.111111) (12,0.090909) (13,0.128205) (14,0.076923) (16,0.069444) (17,0.066176) (18,0.058824) (20,0.060526) (22,0.047619) (23,0.075099) (24,0.068841) (25,0.046667) (28,0.063492) (32,0.032258) (35,0.035294) (37,0.040541) (39,0.026991) (42,0.040650)}; 
\addlegendentry{cycle, $\alpha = 4$};

\addplot+ [color = lime!70!black,only marks, mark=+] coordinates{  (2,1.000000) (3,0.666667) (4,0.333333) (5,0.300000) (6,0.200000) (8,0.142857) (10,0.111111) (12,0.090909) (14,0.076923) (16,0.066667) (17,0.183824) (18,0.058824) (20,0.052632) (21,0.128571) (22,0.099567) (31,0.068817) (35,0.057143) (39,0.048583) (44,0.035941) (55,0.029630)}; 
\addlegendentry{tree, $\alpha = 4$};

\addplot[black, domain = 2:55]{2*x^(-1.1)};
\addlegendentry{ $y \sim k^{-1}$};

\end{axis}
\end{tikzpicture}
}
%\caption*{$\beta = 3$.}
\end{minipage}
\caption{Log-log plot of the local clustering coefficient of nodes of a given degree  in  \PS networks obtained by the best move dynamic for $n=3000$ where $\sigma(x) = 18 x^\alpha$. Blue: results starting from a cycle; green: starting from a random tree. Black line: the function $2/k$.}
\label{plot:avg_CC}
\end{figure} 

 Figure~\ref{plot:avg_CC}\ illustrates the correlation between the node degree and the local clustering coefficient of nodes with the respective degree. 
All plots show that the local clustering coefficient is an inverse function of a node degree. 
In Network Science, a local clustering following the law $\sim k^{-1}$ is considered as an indication of the network's hierarchy that is a fundamental property of many real-world networks\cite{ravasz2003hierarchical}.

\begin{table*}[h!]%
\begin{center}
\renewcommand{\arraystretch}{1.5}
\begin{tabular}{@{}lcccccc@{}}
  & \textbf{SNCG, } & \textbf{SNCG, }   &\textbf{ego} &\textbf{ADVOGATO} &\textbf{HAMSTERSTER} \\
  &$\alpha=2$ &$\alpha=3$   &\textbf{Facebook}~\cite{leskovec2012learning} &\cite{nr-aaai15} &\cite{nr-aaai15}  \\
 \hline
 $|V|$  & 3000 & 3000   &4039 &2280 &1348\\
 $|E|$  &18059 &6019   &88234 &5251 &6642\\
Diameter &8 & 11    &8 &11 &6\\
avg distance &3.69 & 5.17  &3.69 &3.85 &3.2\\
max degree &72 &55   &1045 &148 &273\\
avg degree &12 & 4.013    &43.7 &4.61 &9.85\\
avg CC &0.415 & 0.67   &0.617 &0.2868 &0.54\\
\hline
\end{tabular}
\end{center}
\caption{Comparison of basic structural properties of \PS networks of the SNCG and real-world social networks. The networks \textbf{ego-Facebook}, \textbf{ADVOGATO}, and \textbf{HAMSTERSTER}  are (snippets of) online social networks .}
\label{table:compare_with_real}
\end{table*}
\noindent Table~\ref{table:compare_with_real} shows a comparison of a experimentally generated network with $3000$ nodes for $\alpha = 2$ and $\alpha = 3$, and real-world social networks.

In summary, we conclude from our proof-of-concept experiments that the best move dynamic of the SNCG generates \PS networks that have very similar properties as real-world social networks.

\section{Conclusion}
We introduced the SNCG, a promising game-theoretical model of strategic network formation where agents can bilaterally form new connections or unilaterally remove existing links aiming to maximize their centrality in the created network while at the same time to minimize the cost for maintaining edges. We emphasize that our model is based on only four simple principles: (1) agents are selfish, (2) each agent aims at increasing her centrality, (3) new connections are most likely to appear between friends of friends rather than between more remote nodes, and (4) connections are costly and can only be created by bilateral consent. All principles are motivated by modeling real-world social networks. 

For this simple and stylized model for the creation of a social network by selfish agents, we provide theoretical as well as promising empirical results.
On the theory side, we show that equilibrium networks of the SNCG have structural properties expected from social networks, like a high number of triangles, low diameter, and a low number of isolated 1-degree nodes. Our bounds on the PoA and the PoS show that the cost of closing a triangle essentially determines how inefficient an equilibrium network can be, compared to the social optimum. We emphasize, that all our theoretical results hold for a broad class of convex monotone edge cost functions. 
On the empicial side, we provide proof-of-concept results showing that the best move dynamic of the SNCG converges to equilibrium networks that share fundamental properties with real-world networks, like a power-law degree distribution, a high clustering, and a low diameter.

We see our paper as a promising step towards game-theoretic models that yield networks with all core properties of real-world networks. Future work could systematically study the influence of our model parameters on the obtained network features and on proving that the sequential network creation process indeed converges to real-world like networks with high probability.

\bibliographystyle{abbrv}
\bibliography{distNCG}

\newcommand{\SortNoop}[1]{}
\begin{thebibliography}{10}

\bibitem{Al06}
S.~Albers, S.~Eilts, E.~Even-Dar, Y.~Mansour, and L.~Roditty.
\newblock On nash equilibria for a network creation game.
\newblock In {\em SODA'06}, pages 89--98, 2006.

\bibitem{RHB99}
R.~Albert, H.~Jeong, and A.-L. Barab{\'a}si.
\newblock Internet: Diameter of the world-wide web.
\newblock {\em Nature}, 401(6749):130, 1999.

\bibitem{AM19}
C.~{\`{A}}lvarez and A.~Messegu{\'{e}}.
\newblock On the price of anarchy for high-price links.
\newblock In I.~Caragiannis, V.~S. Mirrokni, and E.~Nikolova, editors, {\em
  {WINE}'19}, pages 316--329, 2019.

\bibitem{ADKTWR}
E.~Anshelevich, A.~Dasgupta, J.~Kleinberg, E.~Tardos, T.~Wexler, and
  T.~Roughgarden.
\newblock The price of stability for network design with fair cost allocation.
\newblock {\em SIAM Journal on Computing}, 38(4):1602--1623, 2008.

\bibitem{BG00}
V.~Bala and S.~Goyal.
\newblock A noncooperative model of network formation.
\newblock {\em Econometrica}, 68(5):1181--1229, 2000.

\bibitem{Bar16}
A.-L. Barab{\'a}si.
\newblock {\em Network science}.
\newblock Cambridge University Press, 2016.

\bibitem{B99}
A.-L. Barab{\'a}si and R.~Albert.
\newblock Emergence of scaling in random networks.
\newblock {\em Science}, 286(5439):509--512, 1999.

\bibitem{BFLM19}
D.~Bil\`{o}, T.~Friedrich, P.~Lenzner, and A.~Melnichenko.
\newblock Geometric network creation games.
\newblock In {\em SPAA'19}, pages 323--332. ACM, 2019.

\bibitem{BL20}
D.~Bil{\`{o}} and P.~Lenzner.
\newblock On the tree conjecture for the network creation game.
\newblock {\em Theory Comput. Syst.}, 64(3):422--443, 2020.

\bibitem{BK11}
M.~Brautbar and M.~Kearns.
\newblock A clustering coefficient network formation game.
\newblock In {\em SAGT'11}, pages 224--235. Springer, 2011.

\bibitem{BKL15}
K.~Bringmann, R.~Keusch, and J.~Lengler.
\newblock Geometric inhomogeneous random graphs.
\newblock {\em Theoretical Computer Science}, 760:35--54, 2019.

\bibitem{BKM00}
A.~Broder, R.~Kumar, F.~Maghoul, P.~Raghavan, S.~Rajagopalan, R.~Stata,
  A.~Tomkins, and J.~Wiener.
\newblock Graph structure in the web.
\newblock {\em Computer Networks}, 33(1):309 -- 320, 2000.

\bibitem{Caccetta92}
L.~Caccetta and W.~F. Smyth.
\newblock Graphs of maximum diameter.
\newblock {\em Discret. Math.}, 102(2):121--141, 1992.

\bibitem{CLMM17}
A.~Chauhan, P.~Lenzner, A.~Melnichenko, and L.~Molitor.
\newblock Selfish network creation with non-uniform edge cost.
\newblock In {\em {SAGT}'17}, pages 160--172, 2017.

\bibitem{CLMM16}
A.~Chauhan, P.~Lenzner, A.~Melnichenko, and M.~M{\"{u}}nn.
\newblock On selfish creation of robust networks.
\newblock In {\em {SAGT'16}}, pages 141--152, 2016.

\bibitem{CJKKM19}
Y.~Chen, S.~Jabbari, M.~J. Kearns, S.~Khanna, and J.~Morgenstern.
\newblock Network formation under random attack and probabilistic spread.
\newblock In {\em {IJCAI}'19}, pages 180--186, 2019.

\bibitem{CL02}
F.~Chung and L.~Lu.
\newblock The average distances in random graphs with given expected degrees.
\newblock {\em PNAS}, 99(25):15879--15882, 2002.

\bibitem{CP05}
J.~Corbo and D.~Parkes.
\newblock The price of selfish behavior in bilateral network formation.
\newblock In {\em PODC'05}, pages 99--107, 2005.

\bibitem{CMadH14}
A.~Cord-Landwehr, A.~M{\"a}cker, and F.~M. auf~der Heide.
\newblock Quality of service in network creation games.
\newblock In {\em WINE'14}, pages 423--428. Springer, 2014.

\bibitem{De07}
E.~D. Demaine, M.~T. Hajiaghayi, H.~Mahini, and M.~Zadimoghaddam.
\newblock The price of anarchy in network creation games.
\newblock {\em ACM Trans. on Algorithms}, 8(2):13, 2012.

\bibitem{EFLM20}
H.~Echzell, T.~Friedrich, P.~Lenzner, and A.~Melnichenko.
\newblock Flow-based network creation games.
\newblock In {\em {IJCAI}'20}, pages 139--145, 2020.

\bibitem{Fab03}
A.~Fabrikant, A.~Luthra, E.~Maneva, C.~H. Papadimitriou, and S.~Shenker.
\newblock On a network creation game.
\newblock PODC'03, pages 347--351, 2003.

\bibitem{FK15}
T.~Friedrich and A.~Krohmer.
\newblock On the diameter of hyperbolic random graphs.
\newblock In {\em ICALP'15}, pages 614--625. Springer, 2015.

\bibitem{GJKKM16}
S.~Goyal, S.~Jabbari, M.~J. Kearns, S.~Khanna, and J.~Morgenstern.
\newblock Strategic network formation with attack and immunization.
\newblock In {\em {WINE}'16}, pages 429--443, 2016.

\bibitem{Gul15}
A.~Guly{\'a}s, J.~J. B{\'\i}r{\'o}, A.~K{\H{o}}r{\"o}si, G.~R{\'e}tv{\'a}ri,
  and D.~Krioukov.
\newblock Navigable networks as nash equilibria of navigation games.
\newblock {\em Nature communications}, 6:7651, 2015.

\bibitem{Jac10}
M.~O. Jackson.
\newblock {\em Social and economic networks}.
\newblock Princeton university press, 2010.

\bibitem{JR05}
M.~O. Jackson and B.~W. Rogers.
\newblock The economics of small worlds.
\newblock {\em Journal of the European Economic Association}, 3(2-3):617--627,
  2005.

\bibitem{JW96}
M.~O. Jackson and A.~Wolinsky.
\newblock A strategic model of social and economic networks.
\newblock {\em Journal of economic theory}, 71(1):44--74, 1996.

\bibitem{K00}
J.~Kleinberg.
\newblock The small-world phenomenon: An algorithmic perspective.
\newblock In {\em STOC'00}, STOC '00, pages 163--170, New York, NY, USA, 2000.
  ACM.

\bibitem{KMPRS15}
G.~Kouroupas, E.~Markakis, C.~Papadimitriou, V.~Rigas, and M.~Sideri.
\newblock The web graph as an equilibrium.
\newblock In {\em SAGT'15}, pages 203--215. Springer, 2015.

\bibitem{KP99}
E.~Koutsoupias and C.~Papadimitriou.
\newblock Worst-case equilibria.
\newblock In {\em STACS'99}, pages 404--413, 1999.

\bibitem{Kri10}
D.~Krioukov, F.~Papadopoulos, M.~Kitsak, A.~Vahdat, and M.~Bogu\~n\'a.
\newblock Hyperbolic geometry of complex networks.
\newblock {\em Phys. Rev. E}, 82:036106, Sep 2010.

\bibitem{LKF05}
J.~Leskovec, J.~Kleinberg, and C.~Faloutsos.
\newblock Graphs over time: densification laws, shrinking diameters and
  possible explanations.
\newblock In {\em SIGKDD'05}, pages 177--187. ACM, 2005.

\bibitem{leskovec2012learning}
J.~Leskovec and J.~J. Mcauley.
\newblock Learning to discover social circles in ego networks.
\newblock In {\em Advances in neural information processing systems}, pages
  539--547, 2012.

\bibitem{MMM13}
A.~Mamageishvili, M.~Mihal\'{a}k, and D.~M\"uller.
\newblock Tree nash equilibria in the network creation game.
\newblock In {\em WAW'13}, volume 8305 of {\em LNCS}, pages 118--129. 2013.

\bibitem{Mantel07}
W.~Mantel.
\newblock Problem 28.
\newblock {\em Wiskundige Opgaven}, 10(60-61):320, 1907.

\bibitem{MMO14}
E.~A. Meirom, S.~Mannor, and A.~Orda.
\newblock Network formation games with heterogeneous players and the internet
  structure.
\newblock In {\em EC'14}, pages 735--752. ACM, 2014.

\bibitem{MMO15}
E.~A. Meirom, S.~Mannor, and A.~Orda.
\newblock Formation games of reliable networks.
\newblock INFOCOM '15, pages 1760--1768. IEEE, 2015.

\bibitem{MS10}
M.~Mihal\'{a}k and J.~C. Schlegel.
\newblock The price of anarchy in network creation games is (mostly) constant.
\newblock SAGT'10, pages 276--287, 2010.

\bibitem{MS96}
D.~Monderer and L.~S. Shapley.
\newblock Potential games.
\newblock {\em Games and Economic Behavior}, 14(1):124 -- 143, 1996.

\bibitem{NBW11}
M.~Newman, A.-L. Barabasi, and D.~J. Watts.
\newblock {\em The structure and dynamics of networks}.
\newblock Princeton University Press, 2011.

\bibitem{Pap01}
C.~H. Papadimitriou.
\newblock Algorithms, games, and the internet.
\newblock In {\em STOC'01}, pages 749--753, 2001.

\bibitem{ravasz2003hierarchical}
E.~Ravasz and A.-L. Barab{\'a}si.
\newblock Hierarchical organization in complex networks.
\newblock {\em Physical review E}, 67(2):026112, 2003.

\bibitem{nr-aaai15}
R.~A. Rossi and N.~K. Ahmed.
\newblock The network data repository with interactive graph analytics and
  visualization.
\newblock In {\em AAAI'15}, 2015.

\end{thebibliography}

\end{document}